%% file: main.tex
\documentclass[twocolumn,11pt,floatfix,aps,pra,superscriptaddress]{revtex4-2}
\pdfoutput=1
\usepackage{geometry}
\usepackage[T1]{fontenc}
\usepackage{wrapfig}
\usepackage{subfig}
\usepackage{booktabs}

\input{preamble}

\input{figures/Tikz/quantum.tikzstyles}
\input{figures/Tikz/quantum.tikzdefs}
\input{figures/Tikz/tikzit.sty}
\usepackage{float}
\usepackage{hyperref}
\usepackage{enumitem}
\usepackage[ruled, vlined, linesnumbered]{algorithm2e}
\setlist[itemize]{itemsep=1pt, topsep=0pt}
\usepackage[framemethod=tikz]{mdframed}
\theoremstyle{definition}
\newmdtheoremenv[
  hidealllines=true,
  leftline=true,
  innerleftmargin=5pt,
  innerrightmargin=5pt,
  innertopmargin=0pt,
]{defi}{Definition}

\theoremstyle{definition}
\newmdtheoremenv[
  backgroundcolor=black!10,
  innerleftmargin=5pt,
  innerrightmargin=5pt,
  innertopmargin=0pt,
]{lem}{Lemma}

\theoremstyle{definition}
\newmdtheoremenv[
  backgroundcolor=blue!10,
  innerleftmargin=10pt,
  innerrightmargin=10pt,
  innertopmargin=0pt,
]{prop}{Proposition}

\theoremstyle{definition}
\newmdtheoremenv[
backgroundcolor=red!10,
  innerleftmargin=5pt,
  innerrightmargin=5pt,
  innertopmargin=0pt,
]{thm}{Theorem}

\begin{document}

\title{Simplified circuit-level decoding using Knill error correction}

\author{Ewan Murphy}
\email{ewan.murphy@inria.fr}
\affiliation{Inria Paris, 48 rue Barrault, 75013 Paris, France}
\affiliation{Quandela, 7 Rue L\'eonard de Vinci, 91300 Massy, France}
\affiliation{Institute for Quantum Computing, University of Waterloo, 200 University Ave W, Waterloo, ON N2L 3G1, Canada}
\affiliation{Perimeter Institute for Theoretical Physics, 31 Caroline St N, Waterloo, ON N2L 2Y5, Canada}
\author{Subhayan Sahu}
\affiliation{Perimeter Institute for Theoretical Physics, 31 Caroline St N, Waterloo, ON N2L 2Y5, Canada}
\author{Michael Vasmer}
\affiliation{Inria Paris, 48 rue Barrault, 75013 Paris, France}
\affiliation{Institute for Quantum Computing, University of Waterloo, 200 University Ave W, Waterloo, ON N2L 3G1, Canada}
\affiliation{Perimeter Institute for Theoretical Physics, 31 Caroline St N, Waterloo, ON N2L 2Y5, Canada}

\begin{abstract}
  Quantum error correction will likely be essential for building a large-scale quantum computer, but it comes with significant requirements at the level of classical control software.
  In particular, a quantum error-correcting code must be supplemented with a fast and accurate classical decoding algorithm.
  Standard techniques for measuring the parity-check operators of a quantum error-correcting code involve repeated measurements, which both increases the amount of data that needs to be processed by the decoder, and changes the nature of the decoding problem.
  Knill error correction is a technique that replaces repeated syndrome measurements with a single round of measurements, but requires an auxiliary logical Bell state.
  Here, we provide a theoretical and numerical investigation into Knill error correction from the perspective of decoding.
  We give a self-contained description of the protocol, prove its fault tolerance under locally decaying (circuit-level) noise, and numerically benchmark its performance for quantum low-density parity-check codes.
  We show analytically and numerically that the time-constrained decoding problem for Knill error correction can be solved using the same decoder used for the simpler code-capacity noise model, illustrating that Knill error correction may alleviate the stringent requirements on classical control required for building a large-scale quantum computer.
\end{abstract}

\maketitle

\section{Introduction}

Quantum error correction (QEC) is widely-believed to be essential for building a large-scale quantum computer capable of implementing quantum algorithms with many layers of gates~\cite{campbell2017roads}.
Numerous techniques from classical error correction have been ported over to the quantum setting, but a key difference is that measuring the parity-checks of a quantum code is itself a noisy process~\cite{gottesman2024surviving}.
This motivates the study of error models incorporating measurement errors such circuit-level noise, which generalise error models that only consider errors on the data qubits, such as code-capacity noise.
A variety of techniques have been proposed to deal with noisy measurements, including Shor~\cite{shor1996faulttolerant}, Steane~\cite{steane1997active}, and Knill~\cite{knill2005quantum} error correction, as well as schemes using flag qubits~\cite{chao2018quantum,chamberland2018flag,chao2020flag}, repeated parity-check measurements~\cite{dennis2002topological,fowler2009highthreshold,kovalev2013fault,gottesman2014faulttolerant}, and combinations of the above~\cite{lai2017faulttolerant,huang2021shor,huang2021constructions}.

A quantum error-correcting code must be supplemented by a classical algorithm called a decoder, whose input is the parity-check measurement outcomes (the syndrome) and whose output is a recovery operator.
Error correction is successful if the combined effect of the noise and recovery operator acts trivially on the encoded information.
In the context of fault-tolerant quantum computing, it is crucial that the decoder is fast enough to keep up with the rate of syndrome extraction.
Otherwise a backlog of syndrome data will build up, leading to an exponential slowdown of the computation~\cite{terhal2015quantum}.

Repeating parity-check measurements changes the nature of the decoding problem, and therefore a decoder that worked well for code-capacity noise might not work at all for more realistic error models such as circuit-level noise.
This can be particularly problematic when a code has been optimised to have good performance for code-capacity noise, as this performance might not carry over to circuit-level noise models~\cite{pacenti2025construction}.
However, certain syndrome extraction techniques such as Knill and Steane error correction have the useful property that the code-capacity decoder can be used for phenomenological and circuit-level noise.
This property stems from the fact that these techniques use auxiliary logical states, which are measured destructively.

In this work, we provide a theoretical and numerical investigation into Knill error correction (Knill EC), which we argue is particularly well-suited to the constraints of fast decoding.
For Knill EC there are two relevant decoding problems, which we call \emph{online} and \emph{offline} decoding.
Online decoding refers to decoding the measured error syndrome, and must be fast in order to avoid the backlog problem.
Offline decoding refers to preparing the auxiliary logical states, and is less constrained as state preparation can be parallelised.

We prove that the code-capacity decoder is sufficient for fault-tolerant error correction subject to locally decaying noise (assuming a supply of auxiliary logical Bell states), complimenting the previous analyses of Knill~\cite{knill2005quantum,knill2004faulttolerant,knill2004faulttoleranta} and Gottesman~\cite{gottesman2024surviving}.
We perform numerical simulations of Knill EC, including state preparation, for surface codes and high-rate codes.
In the second case, we use the lifted product codes from Ref.~\cite{raveendran2022finite}, which were optimised to have good performance for belief propagation decoding against code-capacity noise (without post processing).
Our results show that an identical decoder can be used for online decoding in Knill EC.
Given that belief propagation is well-suited to fast implementation on specialised hardware, this shows that the online decoding can be extremely fast, making Knill EC a compelling option for architectures optimised for fast error correction and logical gates.

To perform our simulations, we develop a modular framework for benchmarking quantum error correction circuits, that may be of independent interest.
Our framework enables the end-to-end simulation of composed fault-tolerant protocols, where each one performs its own decoding.
The tool automatically handles the interaction between the protocols and the propagation of the Pauli frame.

The remainder of this article is structured as follows.
In Sec.~\ref{sec:preliminaries} we provide the necessary background on quantum error correction and Knill EC.
In Sec.~\ref{sec:fault_tolerance_knill}, we prove that Knill EC is fault tolerant under locally decaying noise, and that the online decoding can be done using the code-capacity decoder.
Next, in Sec.~\ref{sec:numerics}, we present our numerical results for Knill EC for surface codes and lifted product codes.
Finally, in Sec.~\ref{sec:discussion} we discuss the implications of our results and directions for future research.

\section{Preliminaries}
\label{sec:preliminaries}

A stabilizer code is the $+1$ eigenspace of an abelian subgroup $\mathcal{S}$ of the $n$-qubit Pauli group that does not contain $-I$~\cite{gottesman1997}.
We summarise a code using the notation $[\![n,k,d]\!]$, where $k = n - \rk \mathcal S$ denotes the number of encoded logical qubits, and $d$ denotes the minimum weight of a non-trivial logical Pauli operator.
A Calderbank-Shor-Steane (CSS) code is a stabilizer code whose stabilizer generators can be chosen to be either $X$-type or $Z$-type, where an $X$-type ($Z$-type) generator is a tensor product of $X$ and $I$ ($Z$ and $I$) operators~\cite{calderbank1996good,steane1996error}.
Of particular interest are the class of quantum low-density parity-check (LDPC) codes~\cite{breuckmann2021quantum}, which are defined by the property that their stabilizer generators have low weight, and each qubit participates in a small number of generators.
Examples of quantum LDPC codes include the surface code~\cite{bravyi1998quantum,dennis2002topological}, as well as lifted product codes~\cite{panteleev2021degenerate}.

In a stabilizer code, decoding proceeds by measuring the eigenvalues of the stabilizer generators using a syndrome extraction circuit, which yields the error syndrome.
A (classical) decoding algorithm takes the error syndrome as input and outputs a recovery operator; error correction is successful if the product of the error and recovery operator is a stabilizer.
The minimum-weight perfect matching (MWPM) algorithm is a well-known efficient decoder for the surface code~\cite{dennis2002topological}, while belief propagation (BP) is a commonly used decoder for general quantum LDPC codes (often supplemented with post-processing such as ordered statistics decoding~\cite{panteleev2021degenerate,roffe2020}).

We consider two error models in this work: code-capacity noise and circuit-level noise.
In code-capacity noise, errors occur only on the data qubits, i.e., the syndrome extraction is noiseless.
In circuit-level noise, errors occur on every component of the error correction circuit: data qubits, auxiliary qubits, gates, state preparations, and measurements.
For details on the specific circuit-level error model we use, see Appendix~\ref{sec:error_model}.

\begin{figure*}[htbp]
  \centering
  \resizebox{\textwidth}{!}{
    \input{Tikz/knill_base_diagram.tikz}
  }
  \caption{\textbf{Knill error correction.} For a code block $D$ encoded in a general \([\![n,k,d]\!]\) stabilizer code, two auxiliary blocks $A$ and $B$ are initialized in $k$ logical Bell pairs. The shaded region shows the circuit for transversal Bell measurement on $D \otimes A$. The measurement outcomes are passed to a decoder $\mathcal{D}$, which deduces a classical recovery $\mathcal{R}$. The corrected logical operator measurements are then computed and used to determine the logical Pauli correction applied to $B$, completing a logical teleportation from $D$ to $B$.}\label{fig:knill_error_correction}
\end{figure*}
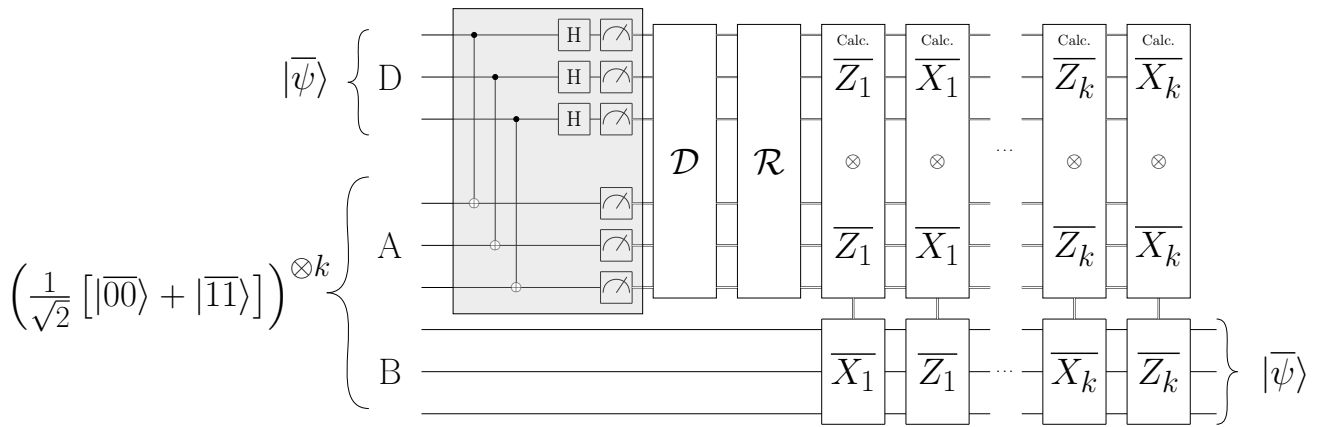

When decoding circuit-level noise for quantum LDPC codes, it is customary to repeat the syndrome extraction circuit $O(d)$ times.
Instead of using the stabilizer generators as parity-checks, one instead uses the difference between the syndromes measured in consecutive rounds, which are often called detectors.
A detector error model is then a Tanner graph~\cite{mackay2003} with detectors as check nodes and error locations as variable nodes, where there an edge between a variable node and a check node if the error at that location would flip the detector~\cite{gidney2021stim}.

\subsection{Knill error correction}
\label{sec:knill_error_correction}

Knill error correction~\cite{knill2005quantum} (Knill EC) is a fault tolerant error correction protocol that uses logical teleportation and a bifurcated decoding strategy. We first describe the protocol, followed by a proof of its fault tolerance.

Suppose that the logical information of $k$ encoded qubits at some stage of a quantum computation is encoded in $n$ physical qubits in a $[\![n,k,d]\!]$ stabilizer code block $D \equiv \{d_i\}_{i = 1}^{n}$, where $d_i$ refer to the $i$th physical qubit. At this stage, an error correction gadget is applied to this state, which usually involves repeated (faulty) syndrome measurement. In Knill EC, we instead consider two auxiliary $[\![n,k,d]\!]$ code blocks: $A \equiv \{a_i\}_{i =1}^{n}$, and $B \equiv \{b_i\}_{i =1}^{n}$, initialized in $k$ logical Bell pair states~\footnote{In principle, the two auxiliary blocks can be in different stabilizer codes, however here we specialize to the same code.}. Next, a logical Bell measurement between the code blocks $D$ and $A$ is performed, using transversal CNOT gates between $D$ and $A$, i.e. $n$ CNOT gates controlled on $d_i$ qubits with targets on $a_i$ qubits, followed by transversal $X$ and $Z$ measurements on $d_i$ and $a_i$ qubits, respectively. 

The transversal CNOT and \(Z\)- and \(X\)-basis measurements is equivalent to measuring \(X_{d_i}\otimes X_{a_i}\) and
\(Z_{d_i}\otimes Z_{a_i}\) for all \(i \in \{1,\dots, n\}\), where
\(X_{a_i}\) is an \(X\) Pauli operator on the \(i\)th physical qubit in block $A$ and similarly for the other operators. These
\(2n\) operators form a basis for all Pauli operators of the form \(P\otimes P\) supported on $D\otimes A$. We can therefore find the eigenvalue for any of these Pauli operators from the \(2n\) bits of measurement data of the transversal Bell measurement. For a stabilizer, the eigenvalue of \(P\otimes P\) will be the XOR of the syndromes in the code blocks $D$ and $A$ for that stabilizer. 
We relay this syndrome information to a decoder $\mathcal D$, which outputs a recovery operator $\mathcal R$.
To find the correct measurement outcomes, we flip the value of the measurement outcomes of \(X_{d_i}\otimes X_{a_i}\) and \(Z_{d_i}\otimes Z_{a_i}\) if $\mathcal R$ anticommutes with $X_{d_i}$ and $Z_{d_i}$, respectively.
We can then deduce the measurement result of the logical operators $\overline{X_{\alpha}}\otimes \overline{X_{\alpha}}$ and $\overline{Z_{\alpha}}\otimes \overline{Z_{\alpha}}$ on $D\otimes A$, where $\alpha = 1, \cdots, k$ denote the encoded logical qubits. Based on the measurement result, we perform logical corrections on block $B$, to teleport the initially encoded information from $D$ to $B$. The protocol is depicted in Fig.~\ref{fig:knill_error_correction}.

Importantly, the decoding necessary for this step of the Knill EC has the same complexity as code-capacity decoding, and thus can be achieved at faster rate than spacetime decoding necessary for repeated syndrome measurements. This constitutes the \textit{online} decoder. In the next section, we will prove that Knill EC is fault tolerant for low enough error in the auxiliary state preparation. However, the preparation of these logical states on $A\otimes B$ with low enough error requires a separate \textit{offline} error correction protocol, which may involve higher decoding complexity than the online decoder.

\section{Fault tolerance with Knill error correction}
\label{sec:fault_tolerance_knill}

We will now analyze logical error propagation under the transversal Bell measurement, which is an essential ingredient of Knill EC.  We will assume each component of the circuit is subjected to locally decaying noise, and study the propagation of logical error through the protocol. Locally decaying (LD) noise means that large, spread-out error events are exponentially unlikely: the chance that errors simultaneously hit any specified set of $t$ qubits (or $t$ fault locations) drops like $\tau^{t}$, with the noise rate being $\tau$.
Equivalently, most of the probability mass is concentrated on small supports, while long correlated patterns are strongly suppressed.
This model still allows correlations, but it rules out ``macroscopic'' correlated failures except with very small probability. The formal definitions are provided in the appendix Sec.~\ref{sec:definitions}. We first present the following lemma,

\begin{lemma}\label{lem:noisy_bell_measurement}
\textbf{Error propagation with transversal Bell measurement circuit. }
  Assume that the blocks $D$ and $A$ are subjected to LD noise with rate $p_1$ and $p_2$ respectively. The transversal CNOT gate $F \subseteq \{1,\dots,n\}$ has LD error with rate $p_3$, and the single-qubit measurements is LD with rate $p_4$. Then, the effective LD noise rate on the two blocks $D$ and $A$ propagated through the faulty Bell measurement circuit is,
  \begin{equation}
    p_{\mathrm{eff}} = \sqrt{p_1 + p_2 + p_3} + p_4.
  \end{equation}
\end{lemma}

We postpone the proof of this lemma to the appendix Sec.~\ref{sec:proof_lemma}. Using this lemma, we find that the final state (before destructive measurement) is effectively erroneous with LD rate $p_{\mathrm{eff}}$ (including measurement errors modelled by bit or phase flips prior to measuring). Hereafter, we can assume perfect syndrome extraction, since the qubits are destructively measured, followed by a faultless classical computation to extract the syndromes. 

We now posit that there exist $[\![n, k, d]\!]$ code families and (code-capacity) decoders based on ideal syndrome measurements, such that for a locally decaying noise channel on the input with error rate $\tau$ below some threshold $\tau_{\text{th}}>0$, the logical failure rate satisfies: 
\begin{align}
    \Pr(\text{logical failure}) \leq f(\tau) \equiv g(\tau, \tau_{\text{th}})\left(\frac{c\tau}{\tau_{\text{th}}}\right)^{d},
\end{align}
for some constant $c< 1$ and a $n$-independent function $g(\tau, \tau_{\text{th}})$. 

We can now prove the fault tolerance of Knill EC under locally decaying noise.
To simplify the notation, we assume that all the associated LD noise rates in the faulty Knill gadget are $p$.
Following Gottesman~\cite{gottesman2024surviving}, we define $c(P,Q)$ be a function that takes as input two Pauli operators $P$ and $Q$, and outputs $1$ if they anti-commute and $0$ if they commute.
\begin{lemma}\label{lem:commutation}
  For any Pauli operators $P,Q,R$, we have $c(P, QR) = c(P, Q) + c(P, R)$.
\end{lemma}
\begin{proof}
  We have 
  \begin{align}
    &PQR = (-1)^{c(P, QR)} QRP, \\ 
    & PQR = (-1)^{c(P, Q)} QPR = (-1)^{c(P, Q) + c(P, R)} QRP.
  \end{align}
  Comparing the exponents, we have the result.
\end{proof}

\begin{figure*}[htbp]
  \centering
  \resizebox{\textwidth}{!}{
    \input{Tikz/fault_tolerant_modules_2.tikz}
  }
  \caption{Clifford fault tolerant protocols can be defined by: a physical circuit, a (classical) function that processes the measurement data sampled from this circuit, and a set of Pauli corrections that are applied based on the function. }\label{fig:fault_tolerant_modules}
\end{figure*}
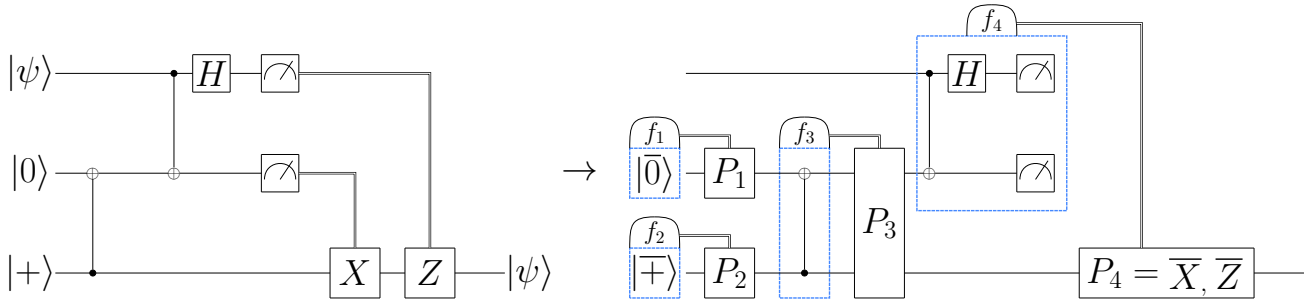

Suppose that the error before the logical Bell measurement is $E_D \otimes E_A$.
Let $F_D$ denote the $Z$ and $Y$ part of the CNOT error on block $D$ along with the measurement error, and let $F_A$ denote the $X$ and $Y$ part of the CNOT error on block $A$ along with the measurement error.
Then, for any Pauli operator $P\otimes P$, the measured value is 
\begin{equation}
\begin{split}
  \tilde m^{P\otimes P} &= m^{P\otimes P} + c(P, E_D) + c(P, E_A) \\ & \quad + c(P, F_D) + c(P, F_A), \\ &= m^{P\otimes P} + c(P, E_D E_A F_D F_A),
\end{split}
\end{equation}
where $m^{P\otimes P}$ is the measurement outcome in the absence of error and we use Lemma~\ref{lem:commutation}.
In particular, for a stabilizer generator $S \otimes S$, the measurement outcome is $\tilde m^{S\otimes S} = c(S, E_D E_A F_D F_A)$, and so we recover the syndrome of the error $E_D E_A F_D F_A$.
If the $p_{\mathrm{eff}}=\sqrt{3p}+p \leq \tau_{\mathrm{th}}$, then with probability $1-f(p_{\mathrm{eff}})$ the recovery operator returned by the code-capacity decoder $\mathcal D$ satisfies $\mathcal R = E_D E_A F_D F_A S$ for some stabilizer $S$.
Then, for any logical Pauli operator $\overline L \otimes \overline L$, we can compute the corrected measurement outcome as
\begin{equation}
\begin{split}
  &\tilde m^{\overline L \otimes \overline L} + c(\overline L, \mathcal R), \\
  &= m^{\overline L \otimes \overline L} + c(\overline L, E_D E_A F_D F_A) + c(\overline L, \mathcal R), \\
  &= m^{\overline L \otimes \overline L} + c(\overline L, S) = m^{\overline L \otimes \overline L}.
\end{split}
\end{equation}
We can therefore recover the correct measurement outcomes for all the logical operators, and thus reliably compute the logical correction to apply to block $B$.

After one round of Knill EC, the new data block has physical error rate equal to the error in the auxiliary blocks in the previous step. There is also a rate of logical failure carried over from the previous rounds. 
Let $p'_{[n]}$ be the probability that the $(n+1)$th round does not eliminate the accumulated logical error from the previous $n$ rounds. The logical error rate after $n+1$ rounds is, thus,
\begin{align}
  q_{[n+1]} &= q_{[n]} \times p'_{[n]} + (1- q_{[n]}) \times f(p_{\mathrm{eff}}) \notag \\
  & \leq q_{[n]} + f(p_{\mathrm{eff}}) 
\end{align}
Therefore, the total logical error probability after $M$ rounds of faulty Knill EC is $q_{[M]} \approx M f(p_{\mathrm{eff}})$. Choosing a large enough distance for the code can suppress $f(\cdot)$, and thus $q_{[M]}$ below a target error rate for an $M\sim \text{poly}(n)$-depth logical computation.

\section{Numerical results}
\label{sec:numerics}

\subsection{Simulating the composition of fault tolerant protocols}\label{sec:modular_simulation_algorithm_brief}
Fault tolerant quantum computing requires the replacement of logical subcircuits with fault tolerant protocols.
Numerical simulations have provided important information on the performance of these protocols.
However, understanding how these protocols perform when composed into larger logical circuits is still an under-explored area.

Motivated by simulating logical Bell teleportation, in the form of Knill EC, we developed a tool that allows the end-to-end simulation of composed fault-tolerant protocols, building on \textsc{Stim}~\cite{gidney2021stim}.
The tool is designed to be modular, with the fault-tolerant protocols defined independently of each other and their interaction handled automatically.
The fault-tolerant protocols we consider are Clifford circuits with an associated classical function that determines what Pauli corrections to apply based on the circuits measurements. An example is shown in Figure~\ref{fig:fault_tolerant_modules}.

Given the modular nature of our tool, each protocol performs its own decoding, which is in contrast with other work in the literature that decodes the whole logical Clifford circuit as one module~\cite{cain2024correlated,wan2024iterative,zhou2025lowoverhead,cain2025fast}.
We note that our approach is also different from overlapping- or sliding-window decoding~\cite{dennis2002topological,skoric2023parallel,berent2024analog,scruby2024highthreshold}, where the decoding problem for the Clifford circuit separated into smaller decoding problems that have overlapping detectors~\cite{bombin2023modular,malcolm2025computing,sahay2025error,zhang2025scalable,serra-peralta2026decoding}.

Although, the original motivation to build this tools was to simulate Knill EC we believe it will be of interest to the wider QEC community as it supports the simulation of many types of logical circuits.
A detailed explanation of how this tool works is given in Appendix~\ref{sec:modular_simulation_algorithm}.

\subsection{Simulations}

We performed simulations to verify that a code-capacity decoder can be used for the online decoder part of Knill EC.
The Knill EC simulations were performed using the modular simulation tool mentioned in the previous section.
This allowed us to have separate the online and offline decoding for the Bell measurement and state preparation, respectively.

To performed these simulations we had to pick a method for state preparation, for this work we chose to use repeated syndrome measurement.
We would like to prepare the state \(\left(\frac{1}{\sqrt{2}}\left[\ket{\overline{00}}+\ket{\overline{11}}\right]\right)^{\otimes k}\) across the \(2k\) logical qubits of two \([\![n,k,d]\!]\) stabilizer code blocks.
We first prepare all the physical qubits of the first code block in the \(\ket{0}\) state, and all the physical qubits of the second code block in the \(\ket{+}\) state.
Then perform \(d\) rounds of stabilizer measurement and decoding, with the \(X\) and \(Z\) syndromes for each state being decoded separately using the offline decoder.
This will prepare the logical states \(\ket{\overline{0}}^{\otimes k}\) and \(\ket{\overline{+}}^{\otimes k}\).
Performing a transversal CNOT controlled on the second code block will create the logical Bell state.
A diagram of this circuit is shown in Figure~\ref{fig:bell_state_preparation}.
This method of state preparation is possible for CSS codes.

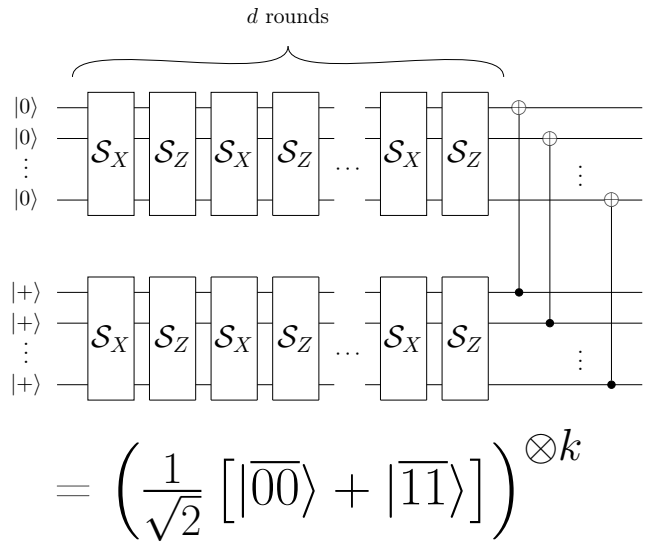
\begin{figure}[htbp]
  \centering
  \resizebox{0.5\textwidth}{!}{
    \input{Tikz/bell_state_preparation.tikz}
  }
  \caption{Circuit for the fault tolerant preparation of logical Bell states used in the simulations. The syndrome data from the \(d\) rounds of \(X\) and \(Z\) stabilizer measurement are decoded separately using the offline decoder.}\label{fig:bell_state_preparation}
\end{figure}

\begin{figure*}[htbp]
    \centering
    
    \subfloat[LP codes, code-capacity noise]{
        \centering
        \includegraphics[width=.45\textwidth]{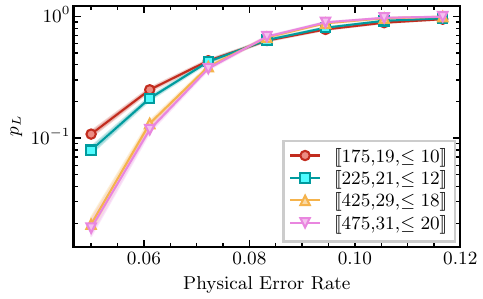}
    }
    \hfill
    \subfloat[Surface codes, code-capacity noise]{
        \centering
        \includegraphics[width=.45\textwidth]{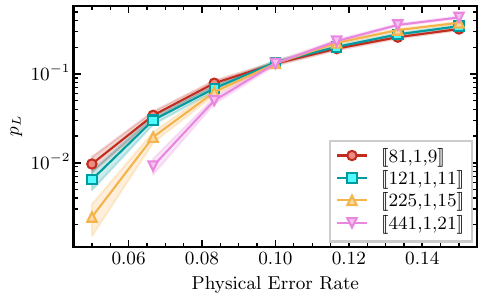}
    }
    \vspace{0.0cm}
    \subfloat[LP codes, circuit noise with Knill EC (8 rounds)]{
        \centering
        \includegraphics[width=.45\textwidth]{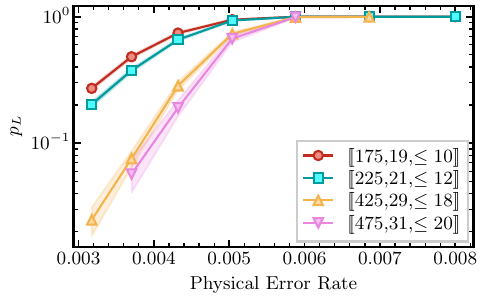}
    }
    \hfill
    \subfloat[Surface codes, circuit noise with Knill EC (8 rounds)]{
        \centering
        \includegraphics[width=.45\textwidth]{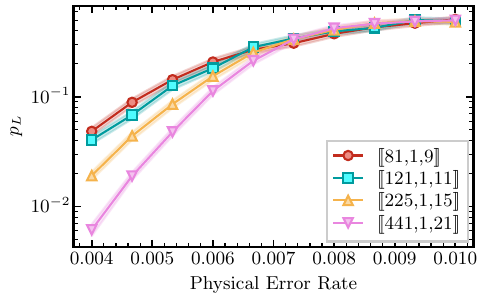}
    }
    \vspace{0.0cm}
    \subfloat[LP codes, circuit noise with repeated syndrome measurement ($32d$ rounds)]{
        \centering
        \includegraphics[width=.45\textwidth]{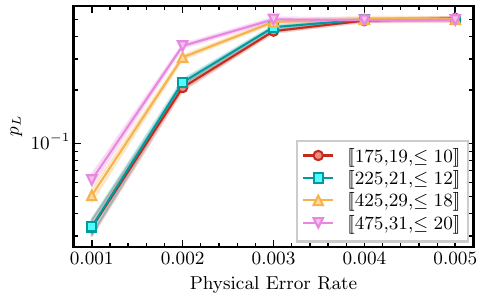}
    }
    \hfill
    \subfloat[Surface codes, circuit noise with repeated syndrome measurement ($32d$ rounds)]{
        \centering
        \includegraphics[width=.45\textwidth]{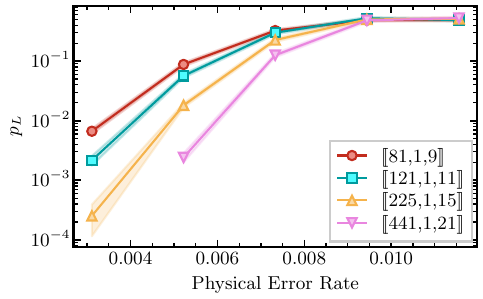}
    }
    \vspace{0.0cm}
    \subfloat[Decoding strategies for numerics. $^{*}$Overlapping window decoding. \textcolor{red}{$^{\dagger}$}No threshold.]{
        \centering
        \resizebox{\textwidth}{!}{
        \begin{tabular}{l|c|cc|c}
            \toprule
            & \textbf{Code-capacity} & \multicolumn{2}{c|}{\textbf{Knill error correction}} & \textbf{Repeated noisy syndrome} \\
            & & \textit{Offline} & \textit{Online} & \textbf{extraction$^*$} \\
            \midrule
            Lifted product codes & BP & BP+OSD & BP & \textcolor{red}{BP$^{\dagger}$} \\
            \hline
            Surface codes        & MWPM & MWPM & MWPM & MWPM \\
            \bottomrule
        \end{tabular}}
    }
    \caption{Numerical simulations comparing the performance of the same decoder for different noise models and error correction gadgets. (a,c,e) Lifted product (LP) codes decoded using BP. We observe results consistent with a non-zero threshold in (a) and (c). (b,d,f) Surface codes decoded using MWPM. We observe results consistent with a non-zero threshold in all cases. Error bars show 95\% confidence intervals calculated using the Agresti-Coull method~\cite{brown2001}. (g) Summary of decoders used in the simulations.}\label{fig:simulation_results}
\end{figure*}

We performed numerical simulations for a set of high rate lifted product (LP) codes from~\cite{raveendran2022finite}; the results are shown in Figure~\ref{fig:simulation_results}(a), (c), and (e).
The LP codes have parameters: \([\![175, 19, \leq 10]]\), \([\![225, 21, \leq 12]\!]\), \([\![425, 29, \leq 18]\!]\) and \([\![475, 31, \leq 20]\!]\), with the two smaller codes having girth 6 and the two larger codes having girth 8.
We performed code-capacity simulations using BP and verified that we were able to observe a threshold.
Then, we performed Knill EC simulations using the same decoder for the online decoding and BP+OSD for the offline decoding.
These simulations confirmed that we could achieve threshold behaviour under circuit level noise using a code-capacity decoder for online decoding, provided we had access to fault tolerantly prepared Bell states.
To provide a comparison with existing methods in the literature we performed overlapping window decoding. 
We simulated \(32d\) rounds of stabilizer measurement with the syndrome partitioned into overlapping windows of size $2d$, with commit size of $d$.
The detector error models for these windows were decoded using BP to compare to the online decoding of Knill EC.
We were not able to achieve threshold behaviour when performing overlapping window decoding with BP for the lifted product codes. This agrees with our understanding that the large girth of these lifted product codes allowed them to be decoded in the code-capacity setting, however, this property is not maintained when moving to the detector error model for repeated syndrome measurement.

In addition to the lifted product codes we performed the same three numerical experiments for the rotated surface code; see Figure~\ref{fig:simulation_results}(b), (d), and (f) for the results.
We performed the simulations on surface codes with parameters \([\![81,1,9]\!]\), \([\![121,1,11]\!]\), \([\![225,1,15]\!]\) and \([\![441,1,21]\!]\). We used MWPM to perform the online and offline decoding.
All three numerical experiments were able to achieve threshold behaviour.
Although MWPM is used for both the online part of Knill EC and overlapping window decoding, the size of the graph they are solving the problem on is different, $O(d^2)$ vertices compared to $O(d^3)$.
This suggests that although overlapping window decoding can achieve threshold behaviour with the same algorithm as online decoding in Knill EC, there still may be advantages in terms of real time decoding and control system design to use Knill EC.

For MWPM we used the PyMatching~\cite{higgott2025} implementation.
For BP and BP+OSD~\cite{roffe2020} we used the implementations of these algorithms from the python library \texttt{ldpc}~\cite{roffe_ldpc_python_tools_2022}.
The parameters for BP were taken from the same paper as the LP codes~\cite{raveendran2022finite}.
Minimum sum was used as the belief propagation method with a scaling factor of 0.75, a serial schedule and maximum number of iterations set at 100.
For BP+OSD the parameters were the same and we used the ``osd\_0'' ordered statistics decoding method.
Table (g) in Figure~\ref{fig:simulation_results} summarises which decoder was used for each of the numerical experiments.

\section{Discussion}
\label{sec:discussion}

In this work, we demonstrated that Knill EC is compatible with fast decoding, and moreover that the online decoding during the protocol can be performed using a code-capacity decoder.
We illustrated this point by using belief propagation (with no post-processing) as the online decoder for a family of lifted product codes.
This decoder is widely used in practice for classical codes, and can be efficiently implemented using specialised hardware such as FPGAs and ASICs.
This fact combined with the small number of variables in the decoding problem implies that the Knill approach has some of the lowest classical control requirements of any quantum error correction scheme.
For decoders compatible with soft information such as matching and belief propagation, it is likely that the performance could be improved by using the circuit structure to optimise the error priors given to the decoder.

Knill EC is compatible with all stabilizer codes, but for CSS codes there is a compressed version that only uses auxiliary logical $|0\rangle$ and $|+\rangle$ states~\cite{huang2021shor,paetznick2024demonstration,baranes2026leveraging}; see Appendix~\ref{app:compressed_knill}.
This version is also compatible with fast decoding and may therefore be preferable for CSS codes.
We argue that compressed Knill EC is likely to be superior to Steane error correction for two reasons.
Firstly, it inherits the robustness to leakage, loss, and coherent errors~\cite{ryan-anderson2024highfidelity,baranes2026leveraging,chang2025taming} of Knill EC.
Secondly, in Steane error correction the recovery operator is applied to the data block, and therefore errors in the auxiliary state can lead to spurious operators being applied to the data block.
This can be dealt with by considering an approach reminiscent of overlapping window decoding~\cite{gong2024improved} but at the cost of making the decoding problem more complex.

Knill EC can be modified to implement logical gates through gate teleportation~\cite{gottesman1999demonstrating,knill2005quantum}.
The only real change is that different auxiliary logical states need to be prepared; the online decoding remains the same.
Therefore, a fault-tolerant quantum computing architecture based on Knill EC could achieve high logical clock speeds, as long as the pace of auxiliary logical state generation can keep up.
This makes architectures based on Knill EC an attractive proposition for hardware platforms with slower physical operations, such as neutral atom and trapped ion qubits~\cite{bluvstein2024logical,muniz2025repeated,chen2024benchmarking,ransford2025helios}.
It is worth contrasting the Knill approach with the paradigm of algorithmic fault tolerance~\cite{zhou2025lowoverhead,cain2025fast,serra-peralta2026decoding}, where syndrome measurements are not repeated but fault tolerance is nonetheless maintained.
This approach is more qubit efficient than the Knill approach, but at the cost of a significant increase in the complexity of the decoding problem.
We leave a detailed comparison of these two approaches to future work.

We benchmarked a naive approach to auxiliary logical state preparation, where logical Bell states are prepared using $d$ rounds of error correction.
Suppose that we want to do Knill EC for $m$ copies of some $[\![n,k,d]\!]$ code.
Then, to keep pace with the online decoding, we would need approximately $2 m d$ code blocks to prepare the auxiliary logical states (for CSS codes we would use compressed Knill and save a factor of $2$), a significant qubit overhead.
Optimised state preparation circuits exist for certain codes, see e.g.~\cite{reichardt2004improved,goswami2023faulttolerant,goswami2024factorybased,gong2024computation,sommers2025observation,kanomata2025faulttolerant,ibe2025measurementbased,yamasaki2024timeefficient,yoshida2025concatenate,litinski2025blocklet}, including code families designed specifically with this in mind~\cite{hastings2025class}.
In future work, we plan to investigate and devise more efficient techniques for logical state preparation, and to propose an optimised architecture for fault-tolerant quantum computation for hardware platforms with access to transversal CNOT gates.

\section*{Code Availability}
The modular simulation tool we built to perform the simulations in this paper is call \textsc{Hex} and is available on GitHub~\cite{murphy_hex_2026}.

\section*{Acknowledgements}
The authors thank Boris Bourdoncle, Anthony Leverrier, and Evan Peters for valuable discussions.
We acknowledge the use of Timo Hillmann's \href{https://github.com/timohillmann/plotting_lib}{\texttt{plotting\_lib}} library for producing the plots in Fig.~\ref{fig:simulation_results}. We acknowledge the use of \href{https://github.com/tikzit/tikzit}{TikZit} to draw the circuit diagrams used in this paper.
The authors are grateful to the CLEPS infrastructure from the Inria of Paris for providing resources and support.
This research was enabled in part by support provided by Simon Fraser University (\url{https://www.sfu.ca/}) and the Digital Research Alliance of Canada (\url{https://alliancecan.ca/}).
EM gratefully acknowledges the financial support of NTT Research while working in Waterloo.
Research at IQC and the Perimeter Institute is supported in part by the Government of Canada through the Department of Innovation, Science and Economic Development; and by the Province of Ontario through the Ministry of Colleges, Universities, Research Excellence and Security.
This work was co-funded by CIFRE grant n°2025/0667.
MV was supported in part by Plan France 2030 through the project ANR-22-PETQ-0006.

\bibliography{bibliography}

\appendix

\section{Proofs}
We provide the detailed definitions and proofs of the lemmas for fault tolerance with Knill error correction (Knill EC) presented earlier.
\subsection{Definitions}\label{sec:definitions}
\begin{definition}
  \textbf{Locally decaying (LD) probability distribution:} For a given finite set $A$, a probability distribution  $p$ over the power set of $A$, $p: \mathrm{Pow}[A] \to [0,1]$ is locally decaying with rate $\tau$ iff for any subset $B \subseteq A$, the total probability that any subset $B' \subseteq A$ drawn according to it is upper bounded as,
\begin{align}
    \sum_{B': B \subseteq B' \subseteq A} p(B') \leq \tau ^{|B|}.
\end{align}
\end{definition}

\begin{definition}
  \textbf{Locally decaying error model on quantum states:} A locally decaying Pauli error model $\mathcal{N}_{\tau}$ with rate $\tau$ on $n$ qubits is a quantum channel specified by Kraus operators consisting of Pauli operators $E \in \hat{P}_n$ drawn with a probability $p_{\mathcal{N}_\tau}(E)$ that satisfies the locally decaying distribution, i.e., the probability distribution of its support is locally decaying with rate $\tau$, 
\begin{align}
    \sum_{\substack{E':~ \text{supp}(E) \subseteq \text{supp}(E')}} p_{\mathcal{N}_\tau}(E') \leq \tau ^{|E|},
\end{align}
where $\text{supp}(E)$ is the support of the non-trivial Pauli factors in $E$ and $|E| = |\text{supp}(E)|$ is its weight. A locally decaying noise channel $\mathcal{N}_{\tau}$ is denoted as $\{\sqrt{p_{\mathcal{N}_\tau}(E)}E\}_{E\in \hat{P}_n}$, where $p_{\mathcal{N}_\tau}$ satisfies the locally decaying condition.
\end{definition}

\begin{definition}
  \textbf{Locally decaying error model for a quantum circuit:} Consider a faulty quantum circuit $Q$, defined as a collection of fault locations (which could be wait locations, multi-qubit gates, and single qubit measurements). By $\mathcal{F}_Q$ we refer to the power set of all fault locations, over which we define a probability distribution $p_{\mathcal{F}}:\mathcal{F}_Q \to [0,1]$. The fault model is locally decaying with rate $\tau$, if the total probability that given fault chain $F\in \mathcal{F}$ is included is upper-bounded by $\tau^{|F|}$, where $|F|$ is the number of fault locations in $F$, and not the support of $F$.
\end{definition}

\subsection{Proof of Lemma~\ref{lem:noisy_bell_measurement}}\label{sec:proof_lemma}
\begin{proof}
For $E \subseteq D \cup A$, define
\[
T(E) := \{ i \in [n] : e_i \cap E \neq \emptyset \},
\]
where $e_i=(d_i,a_i)$ is an ``edge'' (cardinality-two set) corresponding to the CNOT between $d_i$ and $a_i$. Since each edge covers two qubits,
\begin{equation}\label{eq:size_bound}
  |T(E)| \ge \left\lceil \frac{|E|}{2} \right\rceil.
\end{equation}

For qubits in edge $e_i$ to end up in error support on that edge, it must be because either the input qubits to that edge were affected by a prior error, or the edge CNOT was faulty. Let $S_{D}$ and $S_A$ be the support of the effective Pauli operators on $D$ and $A$ blocks. Let $F \subseteq \{1,2,\cdots, n\}$ (edge indices) be the set with faulty CNOT location. Thus, for $E$ to be an error in the output, then for every $i \in T(E)$, either one of these options must be true: $d_i \in S_D$ or $ a_i \in S_A$ or $ i \in F$.

Define
\[
B_i := \underbrace{\{ d_i \in S_D \}}_{1} \cup \underbrace{\{ a_i \in S_A \}}_{2} \cup \underbrace{\{ i \in F \}}_{3}.
\]
Introduce the indicator function $\sigma: T(E) \to \{1,2,3\}$ the error type affecting the CNOT edge in $T(E)$. Thus, the probability of having error $E$ is, 
\begin{align}
  &\Pr(E)  = \Pr(\bigcap_{i \in T(E)} B_i) \nonumber \\
  & = \Pr\left(\bigcup_{\vec{\sigma} \in \{1,2,3\}^{T(E)}} \left(\bigcap_{i: \sigma(i) = 1} \{d_i \in S_D\} \cap \notag \right. \right. \\
  & \left. \left. \bigcap_{i: \sigma(i) = 2} \{a_i \in S_A\} \cap \bigcap_{i: \sigma(i) = 3} \{i \in F\}\right)\right) \nonumber \\
  &\le
\sum_{\vec{\sigma}}
\Pr\!\left(
\{ d_i : i \in I_1 \} \subseteq S_D
\right) \times \notag \\
&\Pr\!\left(
\{ a_i : i \in I_2 \} \subseteq S_A
\right)
\Pr\!\left(
I_3 \subseteq F
\right),
\label{eq:union_bound}
\end{align}
where
\begin{align}
I_j &:= \{ i \in T(E) : \sigma(i) = j \},
\notag \\
&j \in \{1,2,3\},
\notag \\
&
|I_1| + |I_2| + |I_3| = |T(E)| .
\end{align}

Using the locally-decaying (LD) bounds,
\begin{align}
\Pr\!\left(
\{ d_i : i \in I_1 \} \subseteq S_D
\right)
&\le p_1^{|I_1|}, \\
\Pr\!\left(
\{ a_i : i \in I_2 \} \subseteq S_A
\right)
&\le p_2^{|I_2|}, \\
\Pr(I_3 \subseteq F)
&\le p_3^{|I_3|}.
\end{align}

Therefore,
\begin{align}
\Pr\!\left(
\bigcap_{i \in T(E)} B_i
\right)
&\le
\sum_{\substack{I_1,I_2,I_3 \\ \text{partition of } T(E)}}
p_1^{|I_1|}
p_2^{|I_2|}
p_3^{|I_3|}
\notag \\
&=
(p_1 + p_2 + p_3)^{|T(E)|}.
\label{eq:multinomial}
\end{align}

Combining~\eqref{eq:multinomial}, and ~\eqref{eq:size_bound}, we obtain
\begin{align}
\Pr(E)
&\le
(p_1 + p_2 + p_3)^{\lceil |E|/2 \rceil}
\notag \\
&\le
(p_1 + p_2 + p_3)^{|E|/2}
\notag \\
&=
\left(\sqrt{p_1 + p_2 + p_3}\right)^{|E|},
\label{eq:final_ld_bound}
\end{align}
where the first inequality holds assuming $p_1 + p_2 + p_3 \leq 1$.

Let $S_{\mathrm{out}}$ denote the Pauli support induced by the input noise and the noisy CNOT. This is further affected by measurement noise. Let $S_M \subseteq D \cup A$ denote the Pauli support induced by
measurement noise, assumed to be locally decaying with rate $p_4$:
\begin{align}
\Pr(E \subseteq S_M)
\le
p_4^{|E|}
\qquad \forall E \subseteq D \cup A .
\end{align}

Using closure of locally-decaying distributions under union,
for any $E \subseteq D \cup A$,
\begin{align}
\Pr(E)
&\le
\sum_{E_1 \cup E_2 = E}
\Pr(E_1 \subseteq S_{\mathrm{out}})
\Pr(E_2 \subseteq S_M)
\notag \\
&\le
\sum_{E_1 \cup E_2 = E}
\left(\sqrt{p_1 + p_2 + p_3}\right)^{|E_1|}
p_3^{|E_2|}
\notag \\
&=
\left(
\sqrt{p_1 + p_2 + p_3}
+
p_4
\right)^{|E|}.
\end{align}

Therefore the effective LD rate on the
data and first auxiliary block is
\begin{align}
\boxed{
p_{\mathrm{eff}}
=
\sqrt{p_1 + p_2 + p_3}
+
p_4 .
}
\label{eq:peff}
\end{align}

\end{proof}

\section{Modularised circuit simulations}\label{sec:modular_simulation_algorithm}

Here we will provide a more detailed description of the modular simulation tools. The goal of this tool is to perform end-to-end simulations of the composition of fault-tolerant protocols, while allowing each protocol to have independent decoding.

Our tool is designed for Clifford fault-tolerant protocols that take the following form: a Clifford circuit, a classical function that processes the measurements of this circuit (usually the decoder) and a set of Pauli corrections applied by the classical function. We focus on these protocols because the Clifford circuits allow for efficient simulation and the corrections being Pauli operators allows for the classical functions to be applied during post-processing.

Our tool supports a wide range of fault-tolerant protocols such as error correction gadgets, logical state preparation, and logical Clifford gates (e.g. lattice surgery~\cite{horsman2012surface}).
But the tool does not support protocols that require non-Clifford gates or Clifford correction.
Including methods to use post-selection, such state distillation, would require a minor modification to the algorithm presented in this section. 
A boolean value will need to be added to the output of the classical function indicating if a sample needs to be post-selected. 
This is a minor change that does not impact the rest of the algorithm, therefore, for the remainder of this section we will not discuss post-selection.

Now we consider how to simulate a sequence of fault-tolerant (Clifford) protocols.
A naive approach would be to implement a direct stabilizer tableau simulation~\cite{aaronson2004improved}, stopping after each protocol to apply the Pauli corrections before moving onto the next protocol.
However, using \textsc{Stim} it is much faster to sample all of the measurements of the circuit at once. The disadvantage is that the measurements of a protocol have been sampled without the corrections from the previous protocols. Therefore, we need a way to update the measurements to what they would be if we had applied the Pauli corrections during the simulation.

This can be done in pre-processing by propagating the Pauli corrections of each protocol through the full circuit to identify which measurements are flipped by these corrections.
This information can then be used to update the measurements of future protocols depending on the previous protocols corrections. As an example consider the circuit in Figure~\ref{fig:fault_tolerant_modules}. The measurements for protocol 3 will be updated if the corrections from protocols 1 and 2 anticommute with its measurements.

We now give a detailed description of our simulation Algorithm~\ref{alg:sampling_modularised_circuit}.
There are a total of \(N\) fault-tolerant protocols, each with a Clifford circuit \(\mathcal{C}_i\), a classical function \(f_i\), and a set of possible Pauli corrections \(\mathcal{P}_i\) where \(i \in \{1,\ldots , N\}\).
There are \(T_i\) independent Pauli corrections in \(\mathcal{P}_i\).
Composing the circuits from all of the protocols gives our full circuit \(\mathcal{C} = \mathcal{C}_1 \circ \mathcal{C}_2 \circ \cdots \circ \mathcal{C}_N\), each protocol has \(M_i\) measurements with the total circuit having \(M=M_1+\cdots + M_N\) measurements.

As discussed earlier, it is straightforward to construct a function \(\mathbf{F}_{\mathcal{C}} : \mathcal{P} \rightarrow \mathbb{F}_{2}^{M}\) that calculates which measurements in \(\mathcal{C}\) are flipped by introducing a new Pauli gate into the circuit, where $\mathcal P$ is the set of all possible Pauli insertions into the circuit.
We note that $\mathcal P_1 \cup \ldots \cup \mathcal P_N \subseteq \mathcal P$.
The map from Pauli insertions to measurement flips is linear, \(\mathbf{F}_{\mathcal{C}}(P \circ P') = \mathbf{F}_{\mathcal{C}}(P) \oplus \mathbf{F}_{\mathcal{C}}(P')\), allowing the pre-computation of measurement flips for the independent Pauli corrections of each protocol.
\(K_i = \mathbf{F}_{\mathcal{C}}(\mathcal{P}_i) \in \mathbb{F}_{2}^{T_i \times M}\) is the matrix that maps each of the Pauli corrections in protocol \(i\) to their measurement flips.

Circuit \(\mathcal{C}\) is sampled \(N_{\text{shots}}\) times producing \(m = [m_1 | m_2 | \cdots | m_N] \in \mathbb{F}_{2}^{N_{\text{shot}} \times M}\), where \(m_i\) is the sampled data from circuit \(\mathcal{C}_i\). The classical function takes these sampled measurements as input and determines which combination of Pauli operators in \(\mathcal{P}_i\) to apply \(f_i : \mathbb{F}_{2}^{N_{\text{shots}} \times M_i} \rightarrow \mathbb{F}_{2}^{N_{\text{shots}} \times T_i}\). The measurements that then need to be flipped to account for the corrections of protocol \(i\) are \(B_i = f_i(m_i) \cdot K_i\). The measurements are updated, \(m \leftarrow m \oplus B_i\), before the procedure moves onto the next protocol. If one of these protocols includes the measurement of a logical observable then its updated measurement values can be used to calculate logical errors. 

\begin{algorithm}[]
  \caption{Sampling Modular Circuit}\label{alg:sampling_modularised_circuit}
\KwIn{\(\{\mathcal{C}_i\}, \{f_i\}, \{\mathcal{P}_i\}\), \(N_{\text{shots}}\)}
\KwOut{Corrected measurements $m$}
\For{$i = 1$ \KwTo $N$}{
Compute the measurement flips caused by the \(T_i\) independent Pauli corrections in \(\mathcal{P}_i\): \(K_i \leftarrow \mathbf{F}_{\mathcal{C}}(\mathcal{P}_i) \in \mathbb{F}_{2}^{T_i \times M}\) \;
}
Sample the measurements of the full circuit, \(\mathcal{C} = \mathcal{C}_{1} \circ \mathcal{C}_{2} \circ \cdots \circ \mathcal{C}_N\), with \(N_{\text{shot}}\) shots to get \(m \leftarrow [m_1 | m_2 | \cdots | m_N] \in \mathbb{F}_{2}^{N_{\text{shot}} \times M}\)\;
\For{$i = 1$ \KwTo $N$}{
    Apply the classical function for protocol \(i\) to its measurement data \(f_i(m_i) \in \mathbb{F}_{2}^{N_{\text{shot}} \times T_i}\) \;
    Find which measurements need updating to account for this protocol's corrections \(B_i \leftarrow f_i(m_i) \cdot K_i \in \mathbb{F}_{2}^{N_{\text{shots}} \times M}\) \;
    Update the measurements \(m \leftarrow m \oplus B_i\) \;
}
\Return{\(m\)}
\end{algorithm}

To provide clarity on what \(\mathcal{C}_i\), \(f_i\) and \(\mathcal{P}_i\) are for different protocols we will give two examples.
First, the transversal Bell measurement used in Knill error correction, shown in Figure~\ref{fig:knill_error_correction}.
For this \(\mathcal{C}_i\) will be the circuit for transversal Bell measurement shown in the shaded region of the figure.
The classical function \(f_i\) will then be the combination of calculating the syndrome, decoding (\(\mathcal{D}\)), applying the recovery operation (\(\mathcal{R}\)), and calculating the logical Pauli eigenvalues.
The set of Pauli corrections \(\mathcal{P}_i\) will be the \(2k\) logical Pauli operators applied to the third code block, resulting in a \(T_i\) value of \(2k\).
Second, we will consider the repeated syndrome measurement used in state preparation, shown in Figure~\ref{fig:bell_state_preparation}.
The circuit \(\mathcal{C}_i\) is the preparation of the physical qubits and then \(d\) rounds of stabilizer measurement.
The classical function \(f_i\) will calculate the detector values from the measurements and decode these using the detector error model and the offline decoder.
The Pauli corrections \(\mathcal{P}_i\) are defined as the error locations in the detector error model.

These examples illustrate that this tool is not limited to decoding using detector error models, but allows a variety of different quantum error correction protocols to be combined and simulated.

\section{Circuit-level error model}\label{sec:error_model}

\begin{figure}[htbp]
  \centering
  \resizebox{0.5\textwidth}{!}{
    \input{Tikz/stabilizer_measurement_noise_model.tikz}
  }
  \caption{Circuit-level error model for state preparation and stabilizer measurement.}\label{fig:stabilizer_measurement_noise_model}
\end{figure}
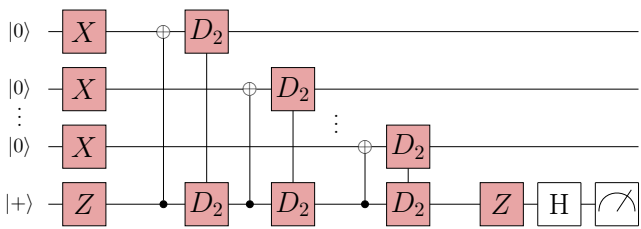

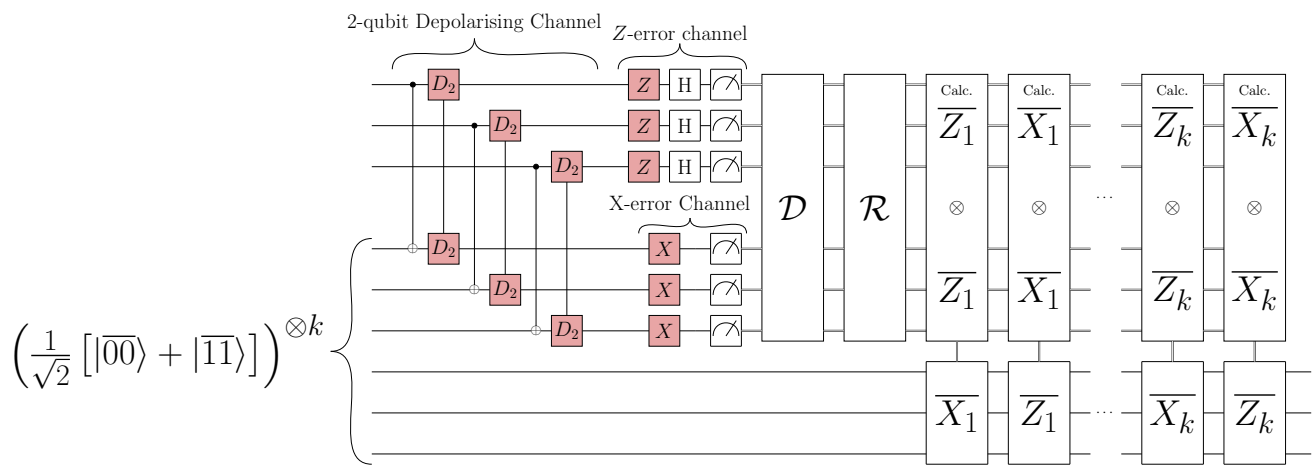
\begin{figure*}[htbp]
  \centering
  \resizebox{\textwidth}{!}{
    \input{Tikz/knill_noise_model.tikz}
  }
  \caption{Circuit-level error model for transversal Bell measurement.}\label{fig:knill_error_correction_noise}
\end{figure*}

We use a circuit-level depolarising error model for our simulations.
Physical qubits are prepared in either the \(\ket{0}\) or \(\ket{+}\) state followed by a bit-flip or phase-flip Pauli error channel, respectively.
Single qubit unitary gates are followed by the 1-qubit depolarising channel.
Two qubit unitary gates are followed by the 2-qubit depolarising channel.
We exclusively consider single-qubit measurements, with the bit-flip and phase-flip channels applied before measuring in the \(Z\) and \(X\) basis, respectively.
Every error channel in the circuit is parametrised by a single error probability \(0 < q < 1\).
Definitions for the error channels in terms of \(q\) are as follows:

\begin{itemize}
  \item Single qubit depolarising channel: \[E_{q}(\rho)=(1-q)\rho + \frac{q}{3}X\rho X + \frac{q}{3}Y\rho Y + \frac{q}{3}Z\rho Z,\]
  \item 2-qubit depolarising channel:
        {\small
        \begin{align*}
          E_{q}(\rho) &= (1-q)\rho\\
           +  \frac{q}{15}\big[&IX\rho IX + IY\rho IY + IZ\rho IZ \\
           +  &XI\rho XI + YI\rho YI + ZI\rho ZI \\
           +  &XX\rho XX + XY\rho XY + XZ\rho XZ \\
           +  &YX\rho YX + YY\rho YY + YZ\rho YZ \\
           +  &ZX\rho ZX + ZY\rho ZY + ZZ\rho ZZ\big], \\
        \end{align*}}
  \item Bit-flip channel: \[E_{q}(\rho)= (1-q)\rho + qX\rho X,\]
  \item Phase-flip channel: \[E_{q}(\rho)= (1-q)\rho + qZ\rho Z.\]
\end{itemize}

Figure~\ref{fig:stabilizer_measurement_noise_model} shows this error model applied to state preparation and stabilizer measurement, and Figure~\ref{fig:knill_error_correction_noise} shows this error model applied to transversal Bell measurement.

\section{Compressed Knill error correction}
\label{app:compressed_knill}

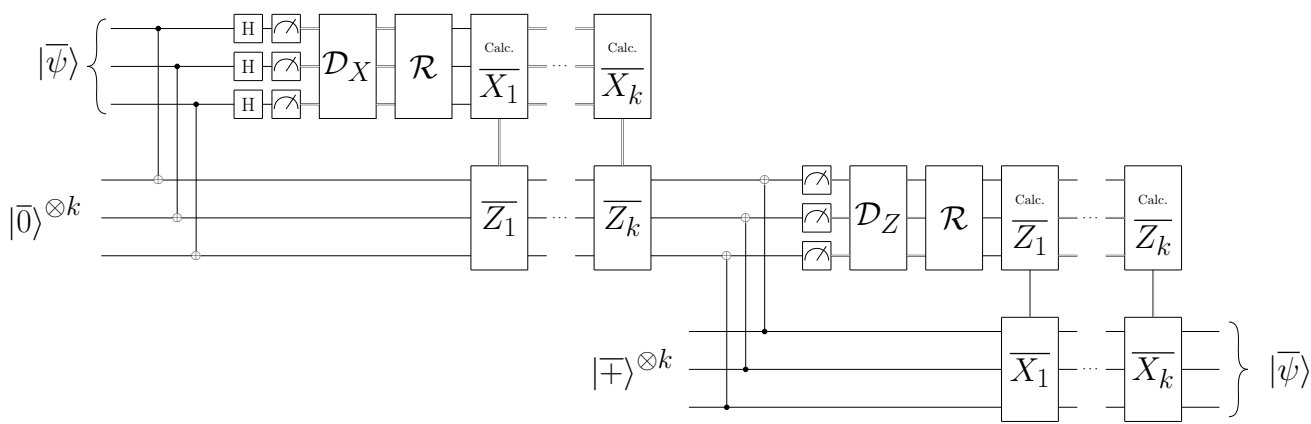
\begin{figure*}[htbp]
  \centering
  \resizebox{\textwidth}{!}{
    \input{Tikz/knill_compressed.tikz}
  }
  \caption{Circuit for the compressed Knill error correction protocol.}\label{fig:knill_compressed}
\end{figure*}

For CSS codes you can perform a compressed form of Knill EC~\cite{huang2021shor,paetznick2024demonstration,baranes2026leveraging}. Instead of using a logical Bell state you use a logical \(\ket{\overline{0}}^{\otimes k}\) and a \(\ket{\overline{+}}^{\otimes k}\) state and perform two teleportations. When teleporting using the \(\ket{\overline{0}}^{\otimes k}\) state the online decoding is done only for the \(X\) stabilizers, similarly, when teleporting using the \(\ket{\overline{+}}^{\otimes k}\) state we only decode using the \(Z\) stabilizers. A diagram of this compressed Knill protocol is shown in Figure~\ref{fig:knill_compressed}.

\end{document}

%% file: preamble.tex
\usepackage{natbib}
\usepackage{graphicx}
\usepackage{mathtools}
\usepackage{amsfonts} 
\usepackage{amsmath} 
\usepackage{subfig}
\usepackage{amssymb}
\usepackage{amsthm}
\usepackage{bm}
\usepackage{braket}
\usepackage{color} 
\usepackage{bbm}
\usepackage{hyperref}
\usepackage{import}
\usepackage[expansion=false]{microtype}
\usepackage{relsize}
\usepackage{framed}
\usepackage[dvipsnames]{xcolor}
\usepackage[percent]{overpic}
\usepackage{tikz}
\usetikzlibrary{positioning}
\usetikzlibrary{calc}

\DeclareMathOperator{\rk}{rk}

\newtheorem*{corollary*}{Corollary}
\newtheorem{lemma}{Lemma}

\theoremstyle{definition}
\newtheorem{definition}{Definition}

\theoremstyle{definition}

%% file: figures/Tikz/quantum.tikzstyles
\tikzstyle{gate}=[shape=rectangle, text height=1.5ex, text depth=0.25ex, yshift=0.5mm, fill=white, draw=black, minimum height=5mm, yshift=-0.5mm, minimum width=5mm, font={\small}, tikzit category=circuit]
\tikzstyle{big gate}=[shape=rectangle, text height=1.5ex, text depth=0.25ex, yshift=0.5mm, fill=white, draw=black, minimum height=10mm, yshift=-0.5mm, minimum width=5mm, font={\small}, tikzit category=circuit]
\tikzstyle{Z dot}=[inner sep=0mm, minimum size=2mm, shape=circle, draw=black, fill=zxgreen, tikzit fill={rgb,255: red,221; green,255; blue,221}, tikzit category=zx]
\tikzstyle{Z bold dot}=[inner sep=0mm, minimum size=2mm, shape=circle, draw=black, fill=zxgreen, tikzit fill={rgb,255: red,221; green,255; blue,221}, line width=1.2pt, tikzit category=zx]
\tikzstyle{Z phase dot}=[minimum size=5mm, font={\footnotesize\boldmath}, shape=rectangle, rounded corners=2mm, inner sep=0.2mm, outer sep=-2mm, scale=0.8, tikzit shape=circle, draw=black, fill=zxgreen, tikzit fill={rgb,255: red,221; green,255; blue,221}, tikzit draw=blue, tikzit category=zx]
\tikzstyle{Z tiny dot}=[inner sep=0mm, minimum size=1mm, shape=circle, draw=black, fill=zxgreen, tikzit fill={rgb,255: red,221; green,255; blue,221}]
\tikzstyle{X dot}=[Z dot, shape=circle, draw=black, fill=zxred, tikzit fill={rgb,255: red,255; green,136; blue,136}, tikzit category=zx]
\tikzstyle{X bold dot}=[inner sep=0mm, minimum size=2mm, shape=circle, draw=black, fill=zxred, tikzit fill={rgb,255: red,255; green,136; blue,136}, line width=1.2pt, tikzit category=zx]
\tikzstyle{X phase dot}=[Z phase dot, tikzit shape=circle, tikzit draw=blue, fill=zxred, tikzit fill={rgb,255: red,255; green,136; blue,136}, font={\footnotesize\boldmath}, tikzit category=zx]
\tikzstyle{X tiny dot}=[inner sep=0mm, minimum size=1mm, shape=circle, draw=black, fill=zxred, tikzit fill={rgb,255: red,255; green,136; blue,136}]
\tikzstyle{hadamard}=[fill=yellow, draw=black, shape=rectangle, inner sep=0.6mm, minimum height=1.5mm, minimum width=1.5mm, tikzit category=zx]
\tikzstyle{paulibox}=[fill={rgb,255: red,221; green,221; blue,255}, draw=black, shape=rectangle, inner sep=0.6mm, minimum height=5mm, minimum width=5mm, font={\footnotesize}, text height=1.5ex, text depth=0.25ex, tikzit category=zx]
\tikzstyle{vertex}=[inner sep=0.2mm, minimum size=1mm, shape=circle, draw=black, fill=black, tikzit category=misc]
\tikzstyle{vertex set}=[inner sep=0.2mm, minimum size=1mm, shape=circle, draw=black, fill=white, font={\footnotesize\boldmath}, tikzit category=misc]
\tikzstyle{small black dot}=[fill=black, draw=black, shape=circle, inner sep=0pt, minimum width=1.2mm, tikzit category=circuit]
\tikzstyle{cnot ctrl}=[fill=black, draw=black, shape=circle, inner sep=0pt, minimum width=1.2mm, tikzit category=circuit]
\tikzstyle{cnot targ}=[fill=white, draw=white, shape=circle, tikzit category=circuit, label={center:$\oplus$}, inner sep=0pt, minimum width=2.1mm, tikzit fill={rgb,255: red,102; green,204; blue,255}, tikzit draw=black]
\tikzstyle{ket}=[fill=white, draw=black, shape=regular polygon, regular polygon sides=3, regular polygon rotate=-30, scale=0.7, inner sep=1pt, tikzit category=circuit, tikzit shape=rectangle, tikzit fill=green]
\tikzstyle{bra}=[fill=white, draw=black, shape=regular polygon, regular polygon sides=3, regular polygon rotate=30, scale=0.7, inner sep=1pt, tikzit category=circuit, tikzit shape=rectangle, tikzit fill=red]
\tikzstyle{scalar}=[shape=rectangle, text height=1.5ex, text depth=0.25ex, yshift=0.5mm, fill=white, draw=black, minimum height=5mm, yshift=-0.5mm, minimum width=5mm, font={\small}]
\tikzstyle{clabel}=[fill=white, draw=none, shape=rectangle, tikzit fill={rgb,255: red,56; green,255; blue,242}, font={\footnotesize}, inner sep=1pt, tikzit category=labels]
\tikzstyle{empty diagram}=[draw={gray!40!white}, dashed, shape=rectangle, minimum width=1cm, minimum height=1cm, tikzit category=misc]
\tikzstyle{amap}=[fill=white, draw=black, shape=NEbox, tikzit category=asymmetric, tikzit fill=yellow, tikzit shape=rectangle]
\tikzstyle{amap conj}=[fill=white, draw=black, shape=NWbox, tikzit category=asymmetric, tikzit fill=green, tikzit shape=rectangle]
\tikzstyle{amap adj}=[fill=white, draw=black, shape=SEbox, tikzit category=asymmetric, tikzit fill=red, tikzit shape=rectangle]
\tikzstyle{amap trans}=[fill=white, draw=black, shape=SWbox, tikzit category=asymmetric, tikzit fill=orange, tikzit shape=rectangle]
\tikzstyle{astate}=[fill=white, draw=black, shape=NEtriangle, tikzit category=asymmetric, tikzit shape=circle, tikzit fill=yellow]
\tikzstyle{astate conj}=[fill=white, draw=black, shape=NWtriangle, tikzit category=asymmetric, tikzit shape=circle, tikzit fill=green]
\tikzstyle{astate adj}=[fill=white, draw=black, shape=SEtriangle, tikzit category=asymmetric, tikzit shape=circle, tikzit fill=red]
\tikzstyle{astate trans}=[fill=white, draw=black, shape=SWtriangle, tikzit category=asymmetric, tikzit shape=circle, tikzit fill=orange]
\tikzstyle{box}=[shape=rectangle, text height=1.5ex, text depth=0.25ex, yshift=0.5mm, fill=white, draw=black, minimum height=5mm, yshift=-0.5mm, minimum width=5mm, font={\small}]
\tikzstyle{medium box}=[shape=rectangle, text height=1.5ex, text depth=0.25ex, yshift=0.5mm, fill=white, draw=black, minimum height=10mm, yshift=-0.5mm, minimum width=5mm, font={\small}]

\tikzstyle{simple}=[-]
\tikzstyle{hadamard edge}=[-, dash pattern=on 2pt off 0.5pt, thick, draw={rgb,255: red,68; green,136; blue,255}]
\tikzstyle{box edge}=[-, dashed, dash pattern=on 2pt off 0.5pt, thick, draw={rgb,255: red,203; green,192; blue,225}]
\tikzstyle{brace edge}=[-, tikzit draw=blue, decorate, decoration={brace,amplitude=1mm,raise=-1mm}]
\tikzstyle{diredge}=[->]
\tikzstyle{double edge}=[-, double]
\tikzstyle{gray edge}=[-, {gray!60!white}]
\tikzstyle{pointer edge}=[->, very thick, gray]
\tikzstyle{boldedge}=[-, line width=1.2pt, shorten <=-0.17mm, shorten >=-0.17mm]
\tikzstyle{bidir edge}=[<->, very thick, draw={rgb,255: red,191; green,191; blue,191}]
\tikzstyle{surface X}=[-, tikzit fill=red, fill=zxred]
\tikzstyle{surface Z}=[-, tikzit fill=green, fill=zxgreen]
\tikzstyle{synthesisable}=[-, fill=green]
\tikzstyle{not synthesisable}=[-, fill={rgb,255: red,148; green,148; blue,148}]
\tikzstyle{white fill}=[-, fill=white]
\tikzstyle{dashed}=[-, dash pattern=on 2pt off 0.5pt, thick, draw={rgb,255: red,68; green,136; blue,255}]
\tikzstyle{Quantum Noise}=[-, tikzit fill=red, fill=red]
\tikzstyle{Classical Noise}=[-, tikzit fill=red, fill=red, line width=1.2pt]
\tikzstyle{Shaded}=[-, fill={rgb,255: red,195; green,195; blue,195}, fill opacity=0.3]

%% file: figures/Tikz/knill_base_diagram.tikz
\begin{tikzpicture}
	\begin{pgfonlayer}{nodelayer}
		\node [style=none] (190) at (30.5, 4.5) {\Huge{$\ket{\overline{\psi}}$}};
		\node [style=none] (299) at (8.5, 6.5) {};
		\node [style=none] (300) at (8.5, 4.5) {};
		\node [style=none] (301) at (8.5, 2.5) {};
		\node [style=none] (318) at (0.5, 20.5) {};
		\node [style=none] (319) at (0.5, 18.5) {};
		\node [style=none] (320) at (0.5, 16.5) {};
		\node [style=none] (321) at (0.5, 21) {};
		\node [style=none] (322) at (0.5, 8) {};
		\node [style=none] (323) at (3.5, 8) {};
		\node [style=none] (324) at (3.5, 21) {};
		\node [font={\Huge}, style=none] (325) at (2, 14.5) {$\mathcal{D}$};
		\node [style=none] (326) at (0.5, 12.5) {};
		\node [style=none] (327) at (0.5, 10.5) {};
		\node [style=none] (328) at (0.5, 8.5) {};
		\node [style=none] (329) at (4.5, 20.5) {};
		\node [style=none] (330) at (4.5, 18.5) {};
		\node [style=none] (331) at (4.5, 16.5) {};
		\node [style=none] (332) at (3.5, 20.5) {};
		\node [style=none] (333) at (3.5, 18.5) {};
		\node [style=none] (334) at (3.5, 16.5) {};
		\node [style=none] (335) at (4.5, 21) {};
		\node [style=none] (336) at (4.5, 8) {};
		\node [style=none] (337) at (7.5, 8) {};
		\node [style=none] (338) at (7.5, 21) {};
		\node [style=none] (339) at (8.5, 20.5) {};
		\node [style=none] (340) at (8.5, 18.5) {};
		\node [style=none] (341) at (8.5, 16.5) {};
		\node [style=none] (342) at (7.5, 20.5) {};
		\node [style=none] (343) at (7.5, 18.5) {};
		\node [style=none] (344) at (7.5, 16.5) {};
		\node [font={\Huge}, style=none] (345) at (6, 14.5) {$\mathcal{R}$};
		\node [style=none] (346) at (4.5, 12.5) {};
		\node [style=none] (347) at (4.5, 10.5) {};
		\node [style=none] (348) at (4.5, 8.5) {};
		\node [style=none] (349) at (3.5, 12.5) {};
		\node [style=none] (350) at (3.5, 10.5) {};
		\node [style=none] (351) at (3.5, 8.5) {};
		\node [style=none] (352) at (8.5, 12.5) {};
		\node [style=none] (353) at (8.5, 10.5) {};
		\node [style=none] (354) at (8.5, 8.5) {};
		\node [style=none] (355) at (7.5, 12.5) {};
		\node [style=none] (356) at (7.5, 10.5) {};
		\node [style=none] (357) at (7.5, 8.5) {};
		\node [style=none] (358) at (27.25, 6.5) {};
		\node [style=none] (359) at (27.25, 4.5) {};
		\node [style=none] (360) at (27.25, 2.5) {};
		\node [style=none] (361) at (27.25, 6.5) {};
		\node [style=none] (362) at (26, 6.5) {};
		\node [style=none] (363) at (26, 4.5) {};
		\node [style=none] (364) at (26, 2.5) {};
		\node [style=none] (365) at (27.25, 4.5) {};
		\node [style=none] (367) at (27.25, 2) {};
		\node [style=none] (368) at (27.25, 7) {};
		\node [style=none] (369) at (28.25, 4.5) {};
		\node [style=none] (413) at (8.5, 6.5) {};
		\node [style=none] (414) at (8.5, 4.5) {};
		\node [style=none] (415) at (8.5, 2.5) {};
		\node [style=none] (609) at (-4, 20.5) {};
		\node [style=none] (610) at (-4, 18.5) {};
		\node [style=none] (611) at (-4, 16.5) {};
		\node [style=none] (616) at (-4, 19.75) {};
		\node [style=none] (617) at (-4, 21.25) {};
		\node [style=none] (618) at (-2.5, 19.75) {};
		\node [style=none] (619) at (-2.5, 21.25) {};
		\node [font={\Large}, style=none] (620) at (-3.25, 20.5) {H};
		\node [style=none] (621) at (-4, 17.75) {};
		\node [style=none] (622) at (-4, 19.25) {};
		\node [style=none] (623) at (-2.5, 17.75) {};
		\node [style=none] (624) at (-2.5, 19.25) {};
		\node [font={\Large}, style=none] (625) at (-3.25, 18.5) {H};
		\node [style=none] (626) at (-4, 15.75) {};
		\node [style=none] (627) at (-4, 17.25) {};
		\node [style=none] (628) at (-2.5, 15.75) {};
		\node [style=none] (629) at (-2.5, 17.25) {};
		\node [font={\Large}, style=none] (630) at (-3.25, 16.5) {H};
		\node [style=none] (631) at (-2, 20.5) {};
		\node [style=none] (632) at (-2.5, 20.5) {};
		\node [style=none] (633) at (-2, 18.5) {};
		\node [style=none] (634) at (-2.5, 18.5) {};
		\node [style=none] (635) at (-2, 16.5) {};
		\node [style=none] (636) at (-2.5, 16.5) {};
		\node [style=none] (637) at (-0.5, 20.5) {};
		\node [style=none] (638) at (-0.5, 18.5) {};
		\node [style=none] (639) at (-0.5, 16.5) {};
		\node [style=none] (662) at (-10.5, 12.5) {};
		\node [style=none] (663) at (-10.5, 10.5) {};
		\node [style=none] (664) at (-10.5, 8.5) {};
		\node [style=none] (669) at (-2, 12.5) {};
		\node [style=none] (670) at (-10.5, 12.5) {};
		\node [style=none] (671) at (-2, 10.5) {};
		\node [style=none] (672) at (-10.5, 10.5) {};
		\node [style=none] (673) at (-2, 8.5) {};
		\node [style=none] (674) at (-10.5, 8.5) {};
		\node [style=none] (675) at (-0.5, 12.5) {};
		\node [style=none] (676) at (-0.5, 8.5) {};
		\node [style=none] (677) at (-0.5, 10.5) {};
		\node [style=none] (698) at (-2, 21.25) {};
		\node [style=none] (699) at (-0.5, 21.25) {};
		\node [style=none] (700) at (-0.5, 19.75) {};
		\node [style=none] (701) at (-2, 19.75) {};
		\node [style=none] (702) at (-1.875, 20.25) {};
		\node [style=none] (703) at (-0.625, 20.25) {};
		\node [style=none] (704) at (-1.25, 20.25) {};
		\node [style=none] (705) at (-0.75, 21) {};
		\node [style=none] (706) at (-2, 19.25) {};
		\node [style=none] (707) at (-0.5, 19.25) {};
		\node [style=none] (708) at (-0.5, 17.75) {};
		\node [style=none] (709) at (-2, 17.75) {};
		\node [style=none] (710) at (-1.875, 18.25) {};
		\node [style=none] (711) at (-0.625, 18.25) {};
		\node [style=none] (712) at (-1.25, 18.25) {};
		\node [style=none] (713) at (-0.75, 19) {};
		\node [style=none] (714) at (-2, 17.25) {};
		\node [style=none] (715) at (-0.5, 17.25) {};
		\node [style=none] (716) at (-0.5, 15.75) {};
		\node [style=none] (717) at (-2, 15.75) {};
		\node [style=none] (718) at (-1.875, 16.25) {};
		\node [style=none] (719) at (-0.625, 16.25) {};
		\node [style=none] (720) at (-1.25, 16.25) {};
		\node [style=none] (721) at (-0.75, 17) {};
		\node [style=none] (722) at (-2, 13.25) {};
		\node [style=none] (723) at (-0.5, 13.25) {};
		\node [style=none] (724) at (-0.5, 11.75) {};
		\node [style=none] (725) at (-2, 11.75) {};
		\node [style=none] (726) at (-1.875, 12.25) {};
		\node [style=none] (727) at (-0.625, 12.25) {};
		\node [style=none] (728) at (-1.25, 12.25) {};
		\node [style=none] (729) at (-0.75, 13) {};
		\node [style=none] (730) at (-2, 11.25) {};
		\node [style=none] (731) at (-0.5, 11.25) {};
		\node [style=none] (732) at (-0.5, 9.75) {};
		\node [style=none] (733) at (-2, 9.75) {};
		\node [style=none] (734) at (-1.875, 10.25) {};
		\node [style=none] (735) at (-0.625, 10.25) {};
		\node [style=none] (736) at (-1.25, 10.25) {};
		\node [style=none] (737) at (-0.75, 11) {};
		\node [style=none] (738) at (-2, 9.25) {};
		\node [style=none] (739) at (-0.5, 9.25) {};
		\node [style=none] (740) at (-0.5, 7.75) {};
		\node [style=none] (741) at (-2, 7.75) {};
		\node [style=none] (742) at (-1.875, 8.25) {};
		\node [style=none] (743) at (-0.625, 8.25) {};
		\node [style=none] (744) at (-1.25, 8.25) {};
		\node [style=none] (745) at (-0.75, 9) {};
		\node [style=none] (746) at (8.5, 7) {};
		\node [style=none] (747) at (8.5, 2) {};
		\node [style=none] (748) at (11.5, 2) {};
		\node [style=none] (749) at (11.5, 7) {};
		\node [font={\Huge}, style=none] (750) at (10, 4.5) {$\overline{X_1}$};
		\node [style=none] (751) at (12.5, 6.5) {};
		\node [style=none] (752) at (12.5, 4.5) {};
		\node [style=none] (753) at (12.5, 2.5) {};
		\node [style=none] (754) at (12.5, 7) {};
		\node [style=none] (755) at (12.5, 2) {};
		\node [style=none] (756) at (15.5, 2) {};
		\node [style=none] (757) at (15.5, 7) {};
		\node [font={\Huge}, style=none] (758) at (14, 4.5) {$\overline{Z_1}$};
		\node [style=none] (759) at (11.5, 6.5) {};
		\node [style=none] (760) at (11.5, 4.5) {};
		\node [style=none] (761) at (11.5, 2.5) {};
		\node [style=none] (762) at (8.5, 21) {};
		\node [style=none] (763) at (8.5, 8) {};
		\node [style=none] (764) at (11.5, 8) {};
		\node [style=none] (765) at (11.5, 21) {};
		\node [font={\Huge}, style=none] (766) at (10, 18.5) {$\overline{Z_1}$};
		\node [style=none] (767) at (10, 7) {};
		\node [style=none] (768) at (10, 8) {};
		\node [style=none] (769) at (14, 7) {};
		\node [style=none] (770) at (14, 8) {};
		\node [style=none] (771) at (12.5, 20.5) {};
		\node [style=none] (772) at (12.5, 18.5) {};
		\node [style=none] (773) at (12.5, 16.5) {};
		\node [style=none] (774) at (11.5, 20.5) {};
		\node [style=none] (775) at (11.5, 18.5) {};
		\node [style=none] (776) at (11.5, 16.5) {};
		\node [style=none] (777) at (12.5, 12.5) {};
		\node [style=none] (778) at (12.5, 10.5) {};
		\node [style=none] (779) at (12.5, 8.5) {};
		\node [style=none] (780) at (11.5, 12.5) {};
		\node [style=none] (781) at (11.5, 10.5) {};
		\node [style=none] (782) at (11.5, 8.5) {};
		\node [style=none] (783) at (12.5, 21) {};
		\node [style=none] (784) at (12.5, 8) {};
		\node [style=none] (785) at (15.5, 8) {};
		\node [style=none] (786) at (15.5, 21) {};
		\node [style=none] (788) at (19, 7) {};
		\node [style=none] (789) at (19, 2) {};
		\node [style=none] (790) at (22, 2) {};
		\node [style=none] (791) at (22, 7) {};
		\node [font={\Huge}, style=none] (792) at (20.5, 4.5) {$\overline{X_k}$};
		\node [style=none] (793) at (23, 6.5) {};
		\node [style=none] (794) at (23, 4.5) {};
		\node [style=none] (795) at (23, 2.5) {};
		\node [style=none] (796) at (23, 7) {};
		\node [style=none] (797) at (23, 2) {};
		\node [style=none] (798) at (26, 2) {};
		\node [style=none] (799) at (26, 7) {};
		\node [font={\Huge}, style=none] (800) at (24.5, 4.5) {$\overline{Z_k}$};
		\node [style=none] (801) at (22, 6.5) {};
		\node [style=none] (802) at (22, 4.5) {};
		\node [style=none] (803) at (22, 2.5) {};
		\node [style=none] (804) at (19, 20.5) {};
		\node [style=none] (805) at (19, 18.5) {};
		\node [style=none] (806) at (19, 16.5) {};
		\node [style=none] (807) at (18, 20.5) {};
		\node [style=none] (808) at (18, 18.5) {};
		\node [style=none] (809) at (18, 16.5) {};
		\node [style=none] (810) at (19, 12.5) {};
		\node [style=none] (811) at (19, 10.5) {};
		\node [style=none] (812) at (19, 8.5) {};
		\node [style=none] (813) at (18, 12.5) {};
		\node [style=none] (814) at (18, 10.5) {};
		\node [style=none] (815) at (18, 8.5) {};
		\node [style=none] (816) at (19, 21) {};
		\node [style=none] (817) at (19, 8) {};
		\node [style=none] (818) at (22, 8) {};
		\node [style=none] (819) at (22, 21) {};
		\node [style=none] (821) at (20.5, 7) {};
		\node [style=none] (822) at (20.5, 8) {};
		\node [style=none] (823) at (24.5, 7) {};
		\node [style=none] (824) at (24.5, 8) {};
		\node [style=none] (825) at (23, 20.5) {};
		\node [style=none] (826) at (23, 18.5) {};
		\node [style=none] (827) at (23, 16.5) {};
		\node [style=none] (828) at (22, 20.5) {};
		\node [style=none] (829) at (22, 18.5) {};
		\node [style=none] (830) at (22, 16.5) {};
		\node [style=none] (831) at (23, 12.5) {};
		\node [style=none] (832) at (23, 10.5) {};
		\node [style=none] (833) at (23, 8.5) {};
		\node [style=none] (834) at (22, 12.5) {};
		\node [style=none] (835) at (22, 10.5) {};
		\node [style=none] (836) at (22, 8.5) {};
		\node [style=none] (837) at (23, 21) {};
		\node [style=none] (838) at (23, 8) {};
		\node [style=none] (839) at (26, 8) {};
		\node [style=none] (840) at (26, 21) {};
		\node [style=none] (842) at (16.5, 20.5) {};
		\node [style=none] (843) at (16.5, 18.5) {};
		\node [style=none] (844) at (16.5, 16.5) {};
		\node [style=none] (845) at (15.5, 20.5) {};
		\node [style=none] (846) at (15.5, 18.5) {};
		\node [style=none] (847) at (15.5, 16.5) {};
		\node [style=none] (848) at (16.5, 12.5) {};
		\node [style=none] (849) at (16.5, 10.5) {};
		\node [style=none] (850) at (16.5, 8.5) {};
		\node [style=none] (851) at (15.5, 12.5) {};
		\node [style=none] (852) at (15.5, 10.5) {};
		\node [style=none] (853) at (15.5, 8.5) {};
		\node [style=none] (854) at (17.25, 15) {$\cdots$};
		\node [style=none] (855) at (16.5, 6.5) {};
		\node [style=none] (856) at (16.5, 4.5) {};
		\node [style=none] (857) at (16.5, 2.5) {};
		\node [style=none] (858) at (15.5, 6.5) {};
		\node [style=none] (859) at (15.5, 4.5) {};
		\node [style=none] (860) at (15.5, 2.5) {};
		\node [style=none] (861) at (19, 6.5) {};
		\node [style=none] (862) at (19, 4.5) {};
		\node [style=none] (863) at (19, 2.5) {};
		\node [style=none] (864) at (18, 6.5) {};
		\node [style=none] (865) at (18, 4.5) {};
		\node [style=none] (866) at (18, 2.5) {};
		\node [style=none] (867) at (17.25, 4.5) {$\cdots$};
		\node [style=none] (871) at (10, 20.25) {Calc.};
		\node [style=none] (873) at (-10.5, 20.5) {};
		\node [style=none] (874) at (-10.5, 18.5) {};
		\node [style=none] (875) at (-10.5, 16.5) {};
		\node [style=none] (876) at (-10.5, 6.5) {};
		\node [style=none] (877) at (-10.5, 4.5) {};
		\node [style=none] (878) at (-10.5, 2.5) {};
		\node [style=none] (879) at (-16, 18.5) {\Huge{$\ket{\overline{\psi}}$}};
		\node [style=none] (880) at (-13, 20.75) {};
		\node [style=none] (881) at (-13, 15.75) {};
		\node [style=none] (882) at (-14, 18.25) {};
		\node [style=none] (883) at (-22.5, 8.25) {\Huge{$\left(\frac{1}{\sqrt{2}}\left[\ket{\overline{00}}+\ket{\overline{11}}\right]\right)^{\otimes k}$}};
		\node [style=none] (884) at (-13, 13.75) {};
		\node [style=none] (885) at (-13, 2.75) {};
		\node [style=none] (886) at (-15, 8.25) {};
		\node [style=none] (887) at (-8, 20.5) {};
		\node [style=none] (888) at (-7, 18.5) {};
		\node [style=none] (889) at (-6, 16.5) {};
		\node [style=none] (890) at (-8, 12.5) {};
		\node [style=none] (891) at (-7, 10.5) {};
		\node [style=none] (892) at (-6, 8.5) {};
		\node [style=cnot ctrl] (893) at (-8, 20.5) {};
		\node [style=cnot ctrl] (894) at (-7, 18.5) {};
		\node [style=cnot ctrl] (895) at (-6, 16.5) {};
		\node [style=cnot targ] (896) at (-8, 12.5) {};
		\node [style=cnot targ] (897) at (-7, 10.5) {};
		\node [style=cnot targ] (898) at (-6, 8.5) {};
		\node [style=none] (899) at (-8, 20.5) {};
		\node [style=none] (900) at (-7, 18.5) {};
		\node [style=none] (901) at (-6, 16.5) {};
		\node [style=none] (902) at (-8, 12.5) {};
		\node [style=none] (903) at (-7, 10.5) {};
		\node [style=none] (904) at (-6, 8.5) {};
		\node [style=none] (905) at (-9, 21.75) {};
		\node [style=none] (907) at (-9, 7.25) {};
		\node [style=none] (908) at (0, 7.25) {};
		\node [style=none] (909) at (0, 21.75) {};
		\node [style=none] (912) at (-12, 18.5) {\Huge D};
		\node [style=none] (913) at (-12, 10.5) {\Huge A};
		\node [style=none] (914) at (-12, 4.5) {\Huge B};
		\node [font={\Huge}, style=none] (915) at (10, 10.5) {$\overline{Z_{1}}$};
		\node [style=none] (916) at (10, 14.5) {\Large $\otimes$};
		\node [font={\Huge}, style=none] (917) at (14, 18.5) {$\overline{X_1}$};
		\node [style=none] (918) at (14, 20.25) {Calc.};
		\node [font={\Huge}, style=none] (919) at (14, 10.5) {$\overline{X_{1}}$};
		\node [style=none] (920) at (14, 14.5) {\Large $\otimes$};
		\node [font={\Huge}, style=none] (921) at (20.5, 18.5) {$\overline{Z_k}$};
		\node [style=none] (922) at (20.5, 20.25) {Calc.};
		\node [font={\Huge}, style=none] (923) at (20.5, 10.5) {$\overline{Z_k}$};
		\node [style=none] (924) at (20.5, 14.5) {\Large $\otimes$};
		\node [font={\Huge}, style=none] (925) at (24.5, 18.5) {$\overline{X_k}$};
		\node [style=none] (926) at (24.5, 20.25) {Calc.};
		\node [font={\Huge}, style=none] (927) at (24.5, 10.5) {$\overline{X_k}$};
		\node [style=none] (928) at (24.5, 14.5) {\Large $\otimes$};
	\end{pgfonlayer}
	\begin{pgfonlayer}{edgelayer}
		\draw (322.center)
			 to (323.center)
			 to (324.center)
			 to (321.center)
			 to cycle;
		\draw [style=double] (332.center) to (329.center);
		\draw [style=double] (333.center) to (330.center);
		\draw [style=double] (334.center) to (331.center);
		\draw [style=none] (338.center)
			 to (335.center)
			 to (336.center)
			 to (337.center)
			 to cycle;
		\draw [style=double] (342.center) to (339.center);
		\draw [style=double] (343.center) to (340.center);
		\draw [style=double] (344.center) to (341.center);
		\draw [style=double] (349.center) to (346.center);
		\draw [style=double] (350.center) to (347.center);
		\draw [style=double] (351.center) to (348.center);
		\draw [style=double] (355.center) to (352.center);
		\draw [style=double] (356.center) to (353.center);
		\draw [style=double] (357.center) to (354.center);
		\draw [style=not synthesisable] (363.center) to (365.center);
		\draw [style=not synthesisable] (364.center) to (360.center);
		\draw (362.center) to (361.center);
		\draw [in=0, out=180] (369.center) to (367.center);
		\draw [in=180, out=0] (368.center) to (369.center);
		\draw (617.center) to (616.center);
		\draw (616.center) to (618.center);
		\draw (618.center) to (619.center);
		\draw (619.center) to (617.center);
		\draw (622.center) to (621.center);
		\draw (621.center) to (623.center);
		\draw (623.center) to (624.center);
		\draw (624.center) to (622.center);
		\draw (627.center) to (626.center);
		\draw (626.center) to (628.center);
		\draw (628.center) to (629.center);
		\draw (629.center) to (627.center);
		\draw (632.center) to (631.center);
		\draw (634.center) to (633.center);
		\draw (636.center) to (635.center);
		\draw (670.center) to (669.center);
		\draw (672.center) to (671.center);
		\draw (674.center) to (673.center);
		\draw (698.center) to (699.center);
		\draw (698.center) to (701.center);
		\draw (701.center) to (700.center);
		\draw (700.center) to (699.center);
		\draw [bend left=90, looseness=1.75] (702.center) to (703.center);
		\draw (705.center) to (704.center);
		\draw (706.center) to (707.center);
		\draw (706.center) to (709.center);
		\draw (709.center) to (708.center);
		\draw (708.center) to (707.center);
		\draw [bend left=90, looseness=1.75] (710.center) to (711.center);
		\draw (713.center) to (712.center);
		\draw (714.center) to (715.center);
		\draw (714.center) to (717.center);
		\draw (717.center) to (716.center);
		\draw (716.center) to (715.center);
		\draw [bend left=90, looseness=1.75] (718.center) to (719.center);
		\draw (721.center) to (720.center);
		\draw (722.center) to (723.center);
		\draw (722.center) to (725.center);
		\draw (725.center) to (724.center);
		\draw (724.center) to (723.center);
		\draw [bend left=90, looseness=1.75] (726.center) to (727.center);
		\draw (729.center) to (728.center);
		\draw (730.center) to (731.center);
		\draw (730.center) to (733.center);
		\draw (733.center) to (732.center);
		\draw (732.center) to (731.center);
		\draw [bend left=90, looseness=1.75] (734.center) to (735.center);
		\draw (737.center) to (736.center);
		\draw (738.center) to (739.center);
		\draw (738.center) to (741.center);
		\draw (741.center) to (740.center);
		\draw (740.center) to (739.center);
		\draw [bend left=90, looseness=1.75] (742.center) to (743.center);
		\draw (745.center) to (744.center);
		\draw (747.center) to (746.center);
		\draw (746.center) to (749.center);
		\draw (749.center) to (748.center);
		\draw (748.center) to (747.center);
		\draw (755.center) to (754.center);
		\draw (754.center) to (757.center);
		\draw (757.center) to (756.center);
		\draw (756.center) to (755.center);
		\draw (759.center) to (751.center);
		\draw (760.center) to (752.center);
		\draw (761.center) to (753.center);
		\draw [style=none] (762.center)
			 to (763.center)
			 to (764.center)
			 to (765.center)
			 to cycle;
		\draw [style=double] (768.center) to (767.center);
		\draw [style=double] (770.center) to (769.center);
		\draw [style=double] (774.center) to (771.center);
		\draw [style=double] (775.center) to (772.center);
		\draw [style=double] (776.center) to (773.center);
		\draw [style=double] (780.center) to (777.center);
		\draw [style=double] (781.center) to (778.center);
		\draw [style=double] (782.center) to (779.center);
		\draw [style=none] (785.center)
			 to (786.center)
			 to (783.center)
			 to (784.center)
			 to cycle;
		\draw (789.center) to (788.center);
		\draw (788.center) to (791.center);
		\draw (791.center) to (790.center);
		\draw (790.center) to (789.center);
		\draw (797.center) to (796.center);
		\draw (796.center) to (799.center);
		\draw (799.center) to (798.center);
		\draw (798.center) to (797.center);
		\draw (801.center) to (793.center);
		\draw (802.center) to (794.center);
		\draw (803.center) to (795.center);
		\draw [style=double] (807.center) to (804.center);
		\draw [style=double] (808.center) to (805.center);
		\draw [style=double] (809.center) to (806.center);
		\draw [style=double] (813.center) to (810.center);
		\draw [style=double] (814.center) to (811.center);
		\draw [style=double] (815.center) to (812.center);
		\draw [style=none] (819.center)
			 to (816.center)
			 to (817.center)
			 to (818.center)
			 to cycle;
		\draw [style=double] (822.center) to (821.center);
		\draw [style=double] (824.center) to (823.center);
		\draw [style=double] (828.center) to (825.center);
		\draw [style=double] (829.center) to (826.center);
		\draw [style=double] (830.center) to (827.center);
		\draw [style=double] (834.center) to (831.center);
		\draw [style=double] (835.center) to (832.center);
		\draw [style=double] (836.center) to (833.center);
		\draw [style=none] (840.center)
			 to (837.center)
			 to (838.center)
			 to (839.center)
			 to cycle;
		\draw [style=double] (845.center) to (842.center);
		\draw [style=double] (846.center) to (843.center);
		\draw [style=double] (847.center) to (844.center);
		\draw [style=double] (851.center) to (848.center);
		\draw [style=double] (852.center) to (849.center);
		\draw [style=double] (853.center) to (850.center);
		\draw (858.center) to (855.center);
		\draw (859.center) to (856.center);
		\draw (860.center) to (857.center);
		\draw (864.center) to (861.center);
		\draw (865.center) to (862.center);
		\draw (866.center) to (863.center);
		\draw [style=double] (637.center) to (318.center);
		\draw [style=double] (638.center) to (319.center);
		\draw [style=double] (639.center) to (320.center);
		\draw [style=double] (675.center) to (326.center);
		\draw [style=double] (677.center) to (327.center);
		\draw [style=double] (676.center) to (328.center);
		\draw (611.center) to (875.center);
		\draw (610.center) to (874.center);
		\draw (609.center) to (873.center);
		\draw (413.center) to (876.center);
		\draw (414.center) to (877.center);
		\draw (415.center) to (878.center);
		\draw [in=-180, out=0] (882.center) to (880.center);
		\draw [in=0, out=-180] (881.center) to (882.center);
		\draw [in=-180, out=0] (886.center) to (884.center);
		\draw [in=0, out=-180] (885.center) to (886.center);
		\draw (893) to (896);
		\draw (894) to (897);
		\draw (895) to (898);
		\draw [style=Shaded] (909.center)
			 to (908.center)
			 to (907.center)
			 to (905.center)
			 to cycle;
	\end{pgfonlayer}
\end{tikzpicture}

%% file: figures/Tikz/fault_tolerant_modules_2.tikz
\begin{tikzpicture}
	\begin{pgfonlayer}{nodelayer}
		\node [style=none] (0) at (0, 10) {};
		\node [style=none] (1) at (2, 10) {};
		\node [style=none] (2) at (0, 4) {};
		\node [style=none] (3) at (2, 4) {};
		\node [style=none] (11) at (0, 10) {};
		\node [style=none] (12) at (1, 11.25) {};
		\node [style=none] (13) at (2, 10) {};
		\node [style=none] (14) at (1, 10.6) {\Large $f_3$};
		\node [style=none] (15) at (-3.75, 13) {};
		\node [style=none] (16) at (-1, 9) {};
		\node [style=none] (17) at (-1, 5) {};
		\node [style=none] (18) at (6.75, 13) {};
		\node [style=none] (19) at (3, 9) {};
		\node [style=none] (20) at (3, 5) {};
		\node [style=none] (21) at (-5, 9) {\huge $\ket{\overline{0}}$};
		\node [style=none] (22) at (-5, 5) {\huge $\ket{\overline{+}}$};
		\node [style=cnot ctrl] (24) at (1, 5) {};
		\node [style=cnot targ] (25) at (1, 9) {};
		\node [style=none] (26) at (-6, 10) {};
		\node [style=none] (27) at (-4, 10) {};
		\node [style=none] (28) at (-6, 8) {};
		\node [style=none] (29) at (-4, 8) {};
		\node [style=none] (30) at (-6, 10) {};
		\node [style=none] (31) at (-5, 11.25) {};
		\node [style=none] (32) at (-4, 10) {};
		\node [style=none] (33) at (-5, 10.6) {\Large $f_1$};
		\node [style=none] (35) at (-6, 6) {};
		\node [style=none] (36) at (-4, 6) {};
		\node [style=none] (37) at (-6, 4) {};
		\node [style=none] (38) at (-4, 4) {};
		\node [style=none] (39) at (-6, 6) {};
		\node [style=none] (40) at (-5, 7.25) {};
		\node [style=none] (41) at (-4, 6) {};
		\node [style=none] (42) at (-5, 6.6) {\Large $f_2$};
		\node [style=none] (60) at (3, 10) {};
		\node [style=none] (61) at (5, 10) {};
		\node [style=none] (62) at (3, 4) {};
		\node [style=none] (63) at (5, 4) {};
		\node [style=none] (64) at (4, 7) {\huge $P_3$};
		\node [style=none] (65) at (9.5, 9) {};
		\node [style=none] (66) at (12, 5) {};
		\node [style=none] (67) at (5, 9) {};
		\node [style=none] (68) at (5, 5) {};
		\node [style=none] (69) at (4, 10) {};
		\node [style=none] (70) at (4, 10.5) {};
		\node [style=none] (71) at (2, 10.5) {};
		\node [style=none] (73) at (-3, 9) {};
		\node [style=none] (74) at (-3, 10) {};
		\node [style=none] (75) at (-1, 10) {};
		\node [style=none] (76) at (-3, 8) {};
		\node [style=none] (77) at (-1, 8) {};
		\node [style=none] (78) at (-2, 9) {\huge $P_1$};
		\node [style=none] (80) at (-1, 9) {};
		\node [style=none] (81) at (-2, 10) {};
		\node [style=none] (82) at (-2, 10.5) {};
		\node [style=none] (83) at (-4, 10.5) {};
		\node [style=none] (95) at (-3, 5) {};
		\node [style=none] (96) at (-3, 6) {};
		\node [style=none] (97) at (-1, 6) {};
		\node [style=none] (98) at (-3, 4) {};
		\node [style=none] (99) at (-1, 4) {};
		\node [style=none] (100) at (-2, 5) {\huge $P_2$};
		\node [style=none] (101) at (-2, 6) {};
		\node [style=none] (102) at (-2, 6.5) {};
		\node [style=none] (103) at (-4, 6.5) {};
		\node [style=none] (105) at (-3.75, 9) {};
		\node [style=none] (106) at (-3.75, 5) {};
		\node [style=none] (107) at (-30, 13) {\huge $\ket{\psi}$};
		\node [style=none] (108) at (-30, 5) {\huge $\ket{+}$};
		\node [style=none] (109) at (-30, 9) {\huge $\ket{0}$};
		\node [style=none] (110) at (-20.75, 9) {};
		\node [style=none] (111) at (-29, 9) {};
		\node [style=none] (112) at (-29, 5) {};
		\node [style=none] (113) at (-18, 5) {};
		\node [style=cnot ctrl] (114) at (-27.5, 5) {};
		\node [style=cnot targ] (115) at (-27.5, 9) {};
		\node [style=none] (116) at (-29, 13) {};
		\node [style=none] (117) at (-23.5, 13) {};
		\node [style=none] (118) at (-20.75, 13) {};
		\node [style=none] (119) at (-19.25, 13) {};
		\node [style=none] (120) at (-20.75, 13.75) {};
		\node [style=none] (121) at (-19.25, 13.75) {};
		\node [style=none] (122) at (-19.25, 12.25) {};
		\node [style=none] (123) at (-20.75, 12.25) {};
		\node [style=none] (124) at (-20.625, 12.75) {};
		\node [style=none] (125) at (-19.375, 12.75) {};
		\node [style=none] (126) at (-20, 12.75) {};
		\node [style=none] (127) at (-19.5, 13.5) {};
		\node [style=none] (128) at (-23.5, 13) {};
		\node [style=none] (129) at (-23.5, 12.25) {};
		\node [style=none] (130) at (-23.5, 13.75) {};
		\node [style=none] (131) at (-22, 12.25) {};
		\node [style=none] (132) at (-22, 13.75) {};
		\node [font={\Large}, style=none] (133) at (-22.75, 13) {\huge $H$};
		\node [style=none] (134) at (-22, 13) {};
		\node [style=none] (135) at (-20.75, 9) {};
		\node [style=none] (136) at (-19.25, 9) {};
		\node [style=none] (137) at (-20.75, 9.75) {};
		\node [style=none] (138) at (-19.25, 9.75) {};
		\node [style=none] (139) at (-19.25, 8.25) {};
		\node [style=none] (140) at (-20.75, 8.25) {};
		\node [style=none] (141) at (-20.625, 8.75) {};
		\node [style=none] (142) at (-19.375, 8.75) {};
		\node [style=none] (143) at (-20, 8.75) {};
		\node [style=none] (144) at (-19.5, 9.5) {};
		\node [style=cnot ctrl] (145) at (-24.25, 13) {};
		\node [style=cnot targ] (146) at (-24.25, 9) {};
		\node [style=none] (147) at (-18, 6) {};
		\node [style=none] (148) at (-18, 4) {};
		\node [style=none] (149) at (-16, 6) {};
		\node [style=none] (150) at (-16, 4) {};
		\node [style=none] (151) at (-17, 5) {\huge $X$};
		\node [style=none] (152) at (-15, 6) {};
		\node [style=none] (153) at (-13, 6) {};
		\node [style=none] (154) at (-15, 4) {};
		\node [style=none] (155) at (-13, 4) {};
		\node [style=none] (156) at (-14, 5) {\huge $Z$};
		\node [style=none] (157) at (-17, 6) {};
		\node [style=none] (158) at (-14, 6) {};
		\node [style=none] (159) at (-17, 9) {};
		\node [style=none] (160) at (-14, 13) {};
		\node [style=none] (161) at (-16, 5) {};
		\node [style=none] (162) at (-15, 5) {};
		\node [style=none] (163) at (-15, 5) {};
		\node [style=none] (164) at (-11, 5) {};
		\node [style=none] (165) at (-13, 5) {};
		\node [style=none] (166) at (-10, 5) {\huge $\ket{\psi}$};
		\node [style=none] (167) at (-8, 9) {\Huge \huge $\rightarrow$};
		\node [style=none] (170) at (9.5, 13) {};
		\node [style=none] (171) at (11, 13) {};
		\node [style=none] (172) at (9.5, 13.75) {};
		\node [style=none] (173) at (11, 13.75) {};
		\node [style=none] (174) at (11, 12.25) {};
		\node [style=none] (175) at (9.5, 12.25) {};
		\node [style=none] (176) at (9.625, 12.75) {};
		\node [style=none] (177) at (10.875, 12.75) {};
		\node [style=none] (178) at (10.25, 12.75) {};
		\node [style=none] (179) at (10.75, 13.5) {};
		\node [style=none] (181) at (6.75, 12.25) {};
		\node [style=none] (182) at (6.75, 13.75) {};
		\node [style=none] (183) at (8.25, 12.25) {};
		\node [style=none] (184) at (8.25, 13.75) {};
		\node [font={\Large}, style=none] (185) at (7.5, 13) {\huge $H$};
		\node [style=none] (186) at (8.25, 13) {};
		\node [style=none] (188) at (11, 9) {};
		\node [style=none] (189) at (9.5, 9.75) {};
		\node [style=none] (190) at (11, 9.75) {};
		\node [style=none] (191) at (11, 8.25) {};
		\node [style=none] (192) at (9.5, 8.25) {};
		\node [style=none] (193) at (9.625, 8.75) {};
		\node [style=none] (194) at (10.875, 8.75) {};
		\node [style=none] (195) at (10.25, 8.75) {};
		\node [style=none] (196) at (10.75, 9.5) {};
		\node [style=cnot ctrl] (197) at (6, 13) {};
		\node [style=cnot targ] (198) at (6, 9) {};
		\node [style=none] (199) at (5.5, 14.5) {};
		\node [style=none] (200) at (11.5, 14.5) {};
		\node [style=none] (201) at (5.5, 7.5) {};
		\node [style=none] (202) at (11.5, 7.5) {};
		\node [style=none] (203) at (7.5, 14.5) {};
		\node [style=none] (204) at (8.5, 15.75) {};
		\node [style=none] (205) at (9.5, 14.5) {};
		\node [style=none] (206) at (8.5, 15.1) {\Large $f_4$};
		\node [style=none] (211) at (12, 6) {};
		\node [style=none] (212) at (19, 6) {};
		\node [style=none] (213) at (12, 4) {};
		\node [style=none] (214) at (19, 4) {};
		\node [style=none] (215) at (16.25, 5) {\huge $\overline{X}$};
		\node [style=none] (218) at (14.5, 6) {};
		\node [style=none] (219) at (14.5, 15) {};
		\node [style=none] (220) at (9.5, 15) {};
		\node [style=none] (221) at (18, 5) {\huge $\overline{Z}$};
		\node [style=none] (222) at (21, 5) {};
		\node [style=none] (223) at (19, 5) {};
		\node [style=none] (224) at (13.75, 5) {\huge $P_4 = $};
		\node [style=none] (225) at (17, 4.25) {\Large $,$};
	\end{pgfonlayer}
	\begin{pgfonlayer}{edgelayer}
		\draw [style=dashed] (0.center)
			 to (1.center)
			 to (3.center)
			 to (2.center)
			 to cycle;
		\draw [in=0, out=90, looseness=1.50] (13.center) to (12.center);
		\draw [in=90, out=-180, looseness=1.50] (12.center) to (11.center);
		\draw (18.center) to (15.center);
		\draw (19.center) to (16.center);
		\draw (17.center) to (20.center);
		\draw (25) to (24);
		\draw [style=dashed] (26.center)
			 to (27.center)
			 to (29.center)
			 to (28.center)
			 to cycle;
		\draw [in=0, out=90, looseness=1.50] (32.center) to (31.center);
		\draw [in=90, out=-180, looseness=1.50] (31.center) to (30.center);
		\draw [style=dashed] (35.center) to (36.center);
		\draw [style=dashed] (36.center) to (38.center);
		\draw [style=dashed] (38.center) to (37.center);
		\draw [style=dashed] (37.center) to (35.center);
		\draw [in=0, out=90, looseness=1.50] (41.center) to (40.center);
		\draw [in=90, out=-180, looseness=1.50] (40.center) to (39.center);
		\draw (63.center) to (62.center);
		\draw (62.center) to (60.center);
		\draw (60.center) to (61.center);
		\draw (61.center) to (63.center);
		\draw (66.center) to (68.center);
		\draw (65.center) to (67.center);
		\draw [style=double edge] (71.center)
			 to (70.center)
			 to (69.center);
		\draw (77.center) to (76.center);
		\draw (76.center) to (74.center);
		\draw (74.center) to (75.center);
		\draw (75.center) to (77.center);
		\draw [style=double edge] (83.center)
			 to (82.center)
			 to (81.center);
		\draw (99.center) to (98.center);
		\draw (98.center) to (96.center);
		\draw (96.center) to (97.center);
		\draw (97.center) to (99.center);
		\draw [style=double edge] (103.center)
			 to (102.center)
			 to (101.center);
		\draw (105.center) to (73.center);
		\draw (106.center) to (95.center);
		\draw (112.center) to (113.center);
		\draw (111.center) to (110.center);
		\draw (114) to (115);
		\draw (116.center) to (117.center);
		\draw (120.center) to (121.center);
		\draw (120.center) to (123.center);
		\draw (123.center) to (122.center);
		\draw (122.center) to (121.center);
		\draw [bend left=90, looseness=1.75] (124.center) to (125.center);
		\draw (127.center) to (126.center);
		\draw (130.center) to (129.center);
		\draw (129.center) to (131.center);
		\draw (131.center) to (132.center);
		\draw (132.center) to (130.center);
		\draw (137.center) to (138.center);
		\draw (137.center) to (140.center);
		\draw (140.center) to (139.center);
		\draw (139.center) to (138.center);
		\draw [bend left=90, looseness=1.75] (141.center) to (142.center);
		\draw (144.center) to (143.center);
		\draw (134.center) to (118.center);
		\draw (145) to (146);
		\draw (150.center) to (148.center);
		\draw (148.center) to (147.center);
		\draw (147.center) to (149.center);
		\draw (149.center) to (150.center);
		\draw (155.center) to (154.center);
		\draw (154.center) to (152.center);
		\draw (152.center) to (153.center);
		\draw (153.center) to (155.center);
		\draw [style=double] (136.center)
			 to (159.center)
			 to (157.center);
		\draw [style=double] (119.center)
			 to (160.center)
			 to (158.center);
		\draw (165.center) to (164.center);
		\draw (161.center) to (163.center);
		\draw (172.center) to (173.center);
		\draw (172.center) to (175.center);
		\draw (175.center) to (174.center);
		\draw (174.center) to (173.center);
		\draw [bend left=90, looseness=1.75] (176.center) to (177.center);
		\draw (179.center) to (178.center);
		\draw (182.center) to (181.center);
		\draw (181.center) to (183.center);
		\draw (183.center) to (184.center);
		\draw (184.center) to (182.center);
		\draw (189.center) to (190.center);
		\draw (189.center) to (192.center);
		\draw (192.center) to (191.center);
		\draw (191.center) to (190.center);
		\draw [bend left=90, looseness=1.75] (193.center) to (194.center);
		\draw (196.center) to (195.center);
		\draw (186.center) to (170.center);
		\draw (197) to (198);
		\draw [style=dashed] (199.center)
			 to (200.center)
			 to (202.center)
			 to (201.center)
			 to cycle;
		\draw [in=0, out=90, looseness=1.50] (205.center) to (204.center);
		\draw [in=90, out=-180, looseness=1.50] (204.center) to (203.center);
		\draw (214.center) to (213.center);
		\draw (213.center) to (211.center);
		\draw (211.center) to (212.center);
		\draw (212.center) to (214.center);
		\draw [style=double edge] (220.center)
			 to (219.center)
			 to (218.center);
		\draw (223.center) to (222.center);
	\end{pgfonlayer}
\end{tikzpicture}

%% file: figures/Tikz/bell_state_preparation.tikz
\begin{tikzpicture}
	\begin{pgfonlayer}{nodelayer}
		\node [style=none] (0) at (-9, 12.25) {$\ket{0}$};
		\node [style=none] (1) at (-9, 11.25) {$\ket{0}$};
		\node [style=none] (2) at (-9, 9.25) {$\ket{0}$};
		\node [style=none] (3) at (-9, 10.5) {$\vdots$};
		\node [style=none] (4) at (-8, 12.25) {};
		\node [style=none] (5) at (-8, 11.25) {};
		\node [style=none] (7) at (-8, 9.25) {};
		\node [style=none] (8) at (-7, 12.25) {};
		\node [style=none] (9) at (-7, 11.25) {};
		\node [style=none] (10) at (-7, 9.25) {};
		\node [style=none] (11) at (-9, 6.25) {$\ket{+}$};
		\node [style=none] (12) at (-9, 5.25) {$\ket{+}$};
		\node [style=none] (13) at (-9, 3.25) {$\ket{+}$};
		\node [style=none] (14) at (-9, 4.5) {$\vdots$};
		\node [style=none] (15) at (-8, 6.25) {};
		\node [style=none] (16) at (-8, 5.25) {};
		\node [style=none] (17) at (-8, 3.25) {};
		\node [style=none] (18) at (-7, 6.25) {};
		\node [style=none] (19) at (-7, 5.25) {};
		\node [style=none] (20) at (-7, 3.25) {};
		\node [style=cnot ctrl] (21) at (7, 6.25) {};
		\node [style=cnot ctrl] (22) at (8, 5.25) {};
		\node [style=cnot targ] (24) at (7, 12.25) {};
		\node [style=cnot targ] (25) at (8, 11.25) {};
		\node [style=cnot targ] (26) at (10, 9.25) {};
		\node [style=none] (27) at (-7, 12.75) {};
		\node [style=none] (28) at (-5.5, 12.75) {};
		\node [style=none] (29) at (-7, 8.75) {};
		\node [style=none] (30) at (-5.5, 8.75) {};
		\node [style=none] (31) at (-6.25, 10.75) {\Large $\mathcal{S}_X$};
		\node [style=none] (32) at (-5, 12.75) {};
		\node [style=none] (33) at (-3.5, 12.75) {};
		\node [style=none] (34) at (-5, 8.75) {};
		\node [style=none] (35) at (-3.5, 8.75) {};
		\node [style=none] (36) at (-4.25, 10.75) {\Large $\mathcal{S}_Z$};
		\node [style=none] (88) at (-7, 6.75) {};
		\node [style=none] (89) at (-5.5, 6.75) {};
		\node [style=none] (90) at (-7, 2.75) {};
		\node [style=none] (91) at (-5.5, 2.75) {};
		\node [style=none] (92) at (-6.25, 4.75) {\Large $\mathcal{S}_X$};
		\node [style=none] (93) at (-5, 6.75) {};
		\node [style=none] (94) at (-3.5, 6.75) {};
		\node [style=none] (95) at (-5, 2.75) {};
		\node [style=none] (96) at (-3.5, 2.75) {};
		\node [style=none] (97) at (-4.25, 4.75) {\Large $\mathcal{S}_Z$};
		\node [style=none] (98) at (-5.5, 12.25) {};
		\node [style=none] (99) at (-5.5, 11.25) {};
		\node [style=none] (100) at (-5.5, 9.25) {};
		\node [style=none] (101) at (-5, 12.25) {};
		\node [style=none] (102) at (-5, 11.25) {};
		\node [style=none] (103) at (-5, 9.25) {};
		\node [style=none] (110) at (-5.5, 6.25) {};
		\node [style=none] (111) at (-5.5, 5.25) {};
		\node [style=none] (112) at (-5.5, 3.25) {};
		\node [style=none] (113) at (-5, 6.25) {};
		\node [style=none] (114) at (-5, 5.25) {};
		\node [style=none] (115) at (-5, 3.25) {};
		\node [style=none] (116) at (-3.5, 12.25) {};
		\node [style=none] (117) at (-3.5, 11.25) {};
		\node [style=none] (118) at (-3.5, 9.25) {};
		\node [style=none] (119) at (-3, 12.25) {};
		\node [style=none] (120) at (-3, 11.25) {};
		\node [style=none] (121) at (-3, 9.25) {};
		\node [style=none] (122) at (-3.5, 6.25) {};
		\node [style=none] (123) at (-3.5, 5.25) {};
		\node [style=none] (124) at (-3.5, 3.25) {};
		\node [style=none] (125) at (-3, 6.25) {};
		\node [style=none] (126) at (-3, 5.25) {};
		\node [style=none] (127) at (-3, 3.25) {};
		\node [style=none] (128) at (-3, 12.75) {};
		\node [style=none] (129) at (-1.5, 12.75) {};
		\node [style=none] (130) at (-3, 8.75) {};
		\node [style=none] (131) at (-1.5, 8.75) {};
		\node [style=none] (132) at (-2.25, 10.75) {\Large $\mathcal{S}_X$};
		\node [style=none] (133) at (-1, 12.75) {};
		\node [style=none] (134) at (0.5, 12.75) {};
		\node [style=none] (135) at (-1, 8.75) {};
		\node [style=none] (136) at (0.5, 8.75) {};
		\node [style=none] (137) at (-0.25, 10.75) {\Large $\mathcal{S}_Z$};
		\node [style=none] (138) at (-3, 6.75) {};
		\node [style=none] (139) at (-1.5, 6.75) {};
		\node [style=none] (140) at (-3, 2.75) {};
		\node [style=none] (141) at (-1.5, 2.75) {};
		\node [style=none] (142) at (-2.25, 4.75) {\Large $\mathcal{S}_X$};
		\node [style=none] (143) at (-1, 6.75) {};
		\node [style=none] (144) at (0.5, 6.75) {};
		\node [style=none] (145) at (-1, 2.75) {};
		\node [style=none] (146) at (0.5, 2.75) {};
		\node [style=none] (147) at (-0.25, 4.75) {\Large $\mathcal{S}_Z$};
		\node [style=none] (148) at (-1.5, 12.25) {};
		\node [style=none] (149) at (-1.5, 11.25) {};
		\node [style=none] (150) at (-1.5, 9.25) {};
		\node [style=none] (151) at (-1, 12.25) {};
		\node [style=none] (152) at (-1, 11.25) {};
		\node [style=none] (153) at (-1, 9.25) {};
		\node [style=none] (154) at (-1.5, 6.25) {};
		\node [style=none] (155) at (-1.5, 5.25) {};
		\node [style=none] (156) at (-1.5, 3.25) {};
		\node [style=none] (157) at (-1, 6.25) {};
		\node [style=none] (158) at (-1, 5.25) {};
		\node [style=none] (159) at (-1, 3.25) {};
		\node [style=none] (160) at (2, 12.25) {};
		\node [style=none] (161) at (2, 11.25) {};
		\node [style=none] (162) at (2, 9.25) {};
		\node [style=none] (163) at (2.5, 12.25) {};
		\node [style=none] (164) at (2.5, 11.25) {};
		\node [style=none] (165) at (2.5, 9.25) {};
		\node [style=none] (166) at (2, 6.25) {};
		\node [style=none] (167) at (2, 5.25) {};
		\node [style=none] (168) at (2, 3.25) {};
		\node [style=none] (169) at (2.5, 6.25) {};
		\node [style=none] (170) at (2.5, 5.25) {};
		\node [style=none] (171) at (2.5, 3.25) {};
		\node [style=none] (172) at (2.5, 12.75) {};
		\node [style=none] (173) at (4, 12.75) {};
		\node [style=none] (174) at (2.5, 8.75) {};
		\node [style=none] (175) at (4, 8.75) {};
		\node [style=none] (176) at (3.25, 10.75) {\Large $\mathcal{S}_X$};
		\node [style=none] (177) at (4.5, 12.75) {};
		\node [style=none] (178) at (6, 12.75) {};
		\node [style=none] (179) at (4.5, 8.75) {};
		\node [style=none] (180) at (6, 8.75) {};
		\node [style=none] (181) at (5.25, 10.75) {\Large $\mathcal{S}_Z$};
		\node [style=none] (182) at (2.5, 6.75) {};
		\node [style=none] (183) at (4, 6.75) {};
		\node [style=none] (184) at (2.5, 2.75) {};
		\node [style=none] (185) at (4, 2.75) {};
		\node [style=none] (186) at (3.25, 4.75) {\Large $\mathcal{S}_X$};
		\node [style=none] (187) at (4.5, 6.75) {};
		\node [style=none] (188) at (6, 6.75) {};
		\node [style=none] (189) at (4.5, 2.75) {};
		\node [style=none] (190) at (6, 2.75) {};
		\node [style=none] (191) at (5.25, 4.75) {\Large $\mathcal{S}_Z$};
		\node [style=none] (192) at (4, 12.25) {};
		\node [style=none] (193) at (4, 11.25) {};
		\node [style=none] (194) at (4, 9.25) {};
		\node [style=none] (195) at (4.5, 12.25) {};
		\node [style=none] (196) at (4.5, 11.25) {};
		\node [style=none] (197) at (4.5, 9.25) {};
		\node [style=none] (198) at (4, 6.25) {};
		\node [style=none] (199) at (4, 5.25) {};
		\node [style=none] (200) at (4, 3.25) {};
		\node [style=none] (201) at (4.5, 6.25) {};
		\node [style=none] (202) at (4.5, 5.25) {};
		\node [style=none] (203) at (4.5, 3.25) {};
		\node [style=none] (204) at (0.5, 12.25) {};
		\node [style=none] (205) at (0.5, 11.25) {};
		\node [style=none] (206) at (0.5, 9.25) {};
		\node [style=none] (207) at (1, 12.25) {};
		\node [style=none] (208) at (1, 11.25) {};
		\node [style=none] (209) at (1, 9.25) {};
		\node [style=none] (210) at (0.5, 6.25) {};
		\node [style=none] (211) at (0.5, 5.25) {};
		\node [style=none] (212) at (0.5, 3.25) {};
		\node [style=none] (213) at (1, 6.25) {};
		\node [style=none] (214) at (1, 5.25) {};
		\node [style=none] (215) at (1, 3.25) {};
		\node [style=none] (216) at (1.5, 10.25) {$\cdots$};
		\node [style=none] (217) at (1.5, 4.25) {$\cdots$};
		\node [style=cnot ctrl] (218) at (10, 3.25) {};
		\node [style=none] (219) at (6, 12.25) {};
		\node [style=none] (220) at (6, 11.25) {};
		\node [style=none] (221) at (6, 9.25) {};
		\node [style=none] (222) at (11, 12.25) {};
		\node [style=none] (223) at (11, 11.25) {};
		\node [style=none] (224) at (11, 9.25) {};
		\node [style=none] (225) at (6, 6.25) {};
		\node [style=none] (226) at (6, 5.25) {};
		\node [style=none] (227) at (6, 3.25) {};
		\node [style=none] (228) at (11, 6.25) {};
		\node [style=none] (229) at (11, 5.25) {};
		\node [style=none] (230) at (11, 3.25) {};
		\node [style=none] (232) at (0.5, 0) {\Huge = \(\left(\frac{1}{\sqrt{2}}\left[\ket{\overline{00}}+\ket{\overline{11}}\right]\right)^{\otimes k}\)};
		\node [style=none] (233) at (-7.5, 13.25) {};
		\node [style=none] (234) at (6.5, 13.25) {};
		\node [style=none] (235) at (-0.5, 14.25) {};
		\node [style=none] (236) at (-0.5, 15.25) {$d$ rounds};
		\node [style=none] (237) at (9, 10.25) {$\vdots$};
		\node [style=none] (238) at (9, 4.25) {$\vdots$};
	\end{pgfonlayer}
	\begin{pgfonlayer}{edgelayer}
		\draw (8.center) to (4.center);
		\draw (9.center) to (5.center);
		\draw (10.center) to (7.center);
		\draw (18.center) to (15.center);
		\draw (19.center) to (16.center);
		\draw (20.center) to (17.center);
		\draw (24) to (21);
		\draw (22) to (25);
		\draw (27.center)
			 to (28.center)
			 to (30.center)
			 to (29.center)
			 to cycle;
		\draw (32.center)
			 to (33.center)
			 to (35.center)
			 to (34.center)
			 to cycle;
		\draw (88.center)
			 to (89.center)
			 to (91.center)
			 to (90.center)
			 to cycle;
		\draw (93.center) to (94.center);
		\draw (94.center) to (96.center);
		\draw (96.center) to (95.center);
		\draw (95.center) to (93.center);
		\draw (101.center) to (98.center);
		\draw (102.center) to (99.center);
		\draw (103.center) to (100.center);
		\draw (113.center) to (110.center);
		\draw (114.center) to (111.center);
		\draw (115.center) to (112.center);
		\draw (119.center) to (116.center);
		\draw (120.center) to (117.center);
		\draw (121.center) to (118.center);
		\draw (125.center) to (122.center);
		\draw (126.center) to (123.center);
		\draw (127.center) to (124.center);
		\draw (128.center)
			 to (129.center)
			 to (131.center)
			 to (130.center)
			 to cycle;
		\draw (133.center) to (134.center);
		\draw (134.center) to (136.center);
		\draw (136.center) to (135.center);
		\draw (135.center) to (133.center);
		\draw (138.center) to (139.center);
		\draw (139.center) to (141.center);
		\draw (141.center) to (140.center);
		\draw (140.center) to (138.center);
		\draw (143.center) to (144.center);
		\draw (144.center) to (146.center);
		\draw (146.center) to (145.center);
		\draw (145.center) to (143.center);
		\draw (151.center) to (148.center);
		\draw (152.center) to (149.center);
		\draw (153.center) to (150.center);
		\draw (157.center) to (154.center);
		\draw (158.center) to (155.center);
		\draw (159.center) to (156.center);
		\draw (163.center) to (160.center);
		\draw (164.center) to (161.center);
		\draw (165.center) to (162.center);
		\draw (169.center) to (166.center);
		\draw (170.center) to (167.center);
		\draw (171.center) to (168.center);
		\draw (172.center)
			 to (173.center)
			 to (175.center)
			 to (174.center)
			 to cycle;
		\draw (177.center) to (178.center);
		\draw (178.center) to (180.center);
		\draw (180.center) to (179.center);
		\draw (179.center) to (177.center);
		\draw (182.center) to (183.center);
		\draw (183.center) to (185.center);
		\draw (185.center) to (184.center);
		\draw (184.center) to (182.center);
		\draw (187.center) to (188.center);
		\draw (188.center) to (190.center);
		\draw (190.center) to (189.center);
		\draw (189.center) to (187.center);
		\draw (195.center) to (192.center);
		\draw (196.center) to (193.center);
		\draw (197.center) to (194.center);
		\draw (201.center) to (198.center);
		\draw (202.center) to (199.center);
		\draw (203.center) to (200.center);
		\draw (207.center) to (204.center);
		\draw (208.center) to (205.center);
		\draw (209.center) to (206.center);
		\draw (213.center) to (210.center);
		\draw (214.center) to (211.center);
		\draw (215.center) to (212.center);
		\draw (218) to (26);
		\draw (222.center) to (219.center);
		\draw (223.center) to (220.center);
		\draw (224.center) to (221.center);
		\draw (228.center) to (225.center);
		\draw (229.center) to (226.center);
		\draw (230.center) to (227.center);
		\draw [in=90, out=-90, looseness=0.50] (235.center) to (233.center);
		\draw [in=90, out=-90, looseness=0.50] (235.center) to (234.center);
	\end{pgfonlayer}
\end{tikzpicture}

%% file: figures/Tikz/stabilizer_measurement_noise_model.tikz
\begin{tikzpicture}
	\begin{pgfonlayer}{nodelayer}
		\node [style=none] (213) at (-9.5, 6.25) {$\ket{0}$};
		\node [style=none] (214) at (-9.5, 4.25) {$\ket{0}$};
		\node [style=none] (215) at (-9.5, 2.25) {$\ket{0}$};
		\node [style=none] (216) at (-9.5, 3.5) {$\vdots$};
		\node [style=none] (217) at (-8.5, 6.25) {};
		\node [style=none] (218) at (-8.5, 4.25) {};
		\node [style=none] (220) at (-8.5, 2.25) {};
		\node [style=none] (221) at (-8, 6.25) {};
		\node [style=none] (222) at (-8, 4.25) {};
		\node [style=none] (223) at (-8, 2.25) {};
		\node [style=none] (224) at (-9.5, 0.25) {$\ket{+}$};
		\node [style=none] (225) at (-6.5, 0.25) {};
		\node [style=none] (226) at (-6.5, 6.25) {};
		\node [style=none] (227) at (-6.5, 4.25) {};
		\node [style=none] (228) at (-6.5, 2.25) {};
		\node [font={\Large}, style=none] (229) at (-7.25, 6.25) {$X$};
		\node [font={\Large}, style=none] (230) at (-7.25, 4.25) {$X$};
		\node [style=none] (231) at (-8, 3) {};
		\node [style=none] (232) at (-6.5, 3) {};
		\node [style=none] (233) at (-8, 1.5) {};
		\node [style=none] (234) at (-6.5, 1.5) {};
		\node [font={\Large}, style=none] (235) at (-7.25, 2.25) {$X$};
		\node [style=none] (239) at (-8, 5) {};
		\node [style=none] (240) at (-6.5, 5) {};
		\node [style=none] (241) at (-8, 3.5) {};
		\node [style=none] (242) at (-6.5, 3.5) {};
		\node [style=none] (243) at (-6.5, 5.5) {};
		\node [style=none] (244) at (-8, 5.5) {};
		\node [style=none] (245) at (-8, 7) {};
		\node [style=none] (246) at (-6.5, 7) {};
		\node [style=none] (247) at (-3.75, 6.25) {};
		\node [style=none] (248) at (-0.75, 4.25) {};
		\node [style=none] (249) at (3.25, 2.25) {};
		\node [style=none] (250) at (-3.75, 0.25) {};
		\node [style=none] (251) at (-4.5, 6.25) {};
		\node [style=none] (252) at (-1.5, 4.25) {};
		\node [style=none] (253) at (2.5, 2.25) {};
		\node [style=cnot ctrl] (254) at (-4.5, 0.25) {};
		\node [style=cnot ctrl] (255) at (-1.5, 0.25) {};
		\node [style=cnot ctrl] (256) at (2.5, 0.25) {};
		\node [style=cnot targ] (257) at (-4.5, 6.25) {};
		\node [style=cnot targ] (258) at (-1.5, 4.25) {};
		\node [style=cnot targ] (259) at (2.5, 2.25) {};
		\node [style=none] (260) at (-2.25, 6.25) {};
		\node [style=none] (261) at (0.75, 4.25) {};
		\node [style=none] (262) at (4.75, 2.25) {};
		\node [style=none] (263) at (-2.25, 0.25) {};
		\node [style=none] (264) at (0.75, 0.25) {};
		\node [style=none] (265) at (4.75, 0.25) {};
		\node [style=none] (266) at (-3.75, 7) {};
		\node [style=none] (267) at (-3.75, 5.5) {};
		\node [style=none] (268) at (-2.25, 5.5) {};
		\node [style=none] (269) at (-2.25, 7) {};
		\node [font={\Large}, style=none] (270) at (-3, 6.25) {$D_{2}$};
		\node [style=none] (273) at (-0.75, 5) {};
		\node [style=none] (274) at (0.75, 5) {};
		\node [style=none] (275) at (-0.75, 3.5) {};
		\node [style=none] (276) at (0.75, 3.5) {};
		\node [font={\Large}, style=none] (279) at (0, 4.25) {$D_{2}$};
		\node [style=none] (280) at (3.25, 3) {};
		\node [style=none] (281) at (4.75, 3) {};
		\node [style=none] (282) at (3.25, 1.5) {};
		\node [style=none] (283) at (4.75, 1.5) {};
		\node [font={\Large}, style=none] (286) at (4, 2.25) {$D_{2}$};
		\node [font={\Large}, style=none] (290) at (-3, 0.25) {$D_{2}$};
		\node [font={\Large}, style=none] (295) at (0, 0.25) {$D_{2}$};
		\node [style=none] (296) at (3.25, 1) {};
		\node [style=none] (297) at (4.75, 1) {};
		\node [style=none] (298) at (3.25, -0.5) {};
		\node [style=none] (299) at (4.75, -0.5) {};
		\node [font={\Large}, style=none] (302) at (4, 0.25) {$D_{2}$};
		\node [style=none] (304) at (-0.75, 0.25) {};
		\node [style=none] (305) at (3.25, 0.25) {};
		\node [style=none] (306) at (-0.75, 1) {};
		\node [style=none] (307) at (0.75, 1) {};
		\node [style=none] (308) at (-0.75, -0.5) {};
		\node [style=none] (309) at (0.75, -0.5) {};
		\node [style=none] (310) at (-2.25, -0.5) {};
		\node [style=none] (311) at (-3.75, -0.5) {};
		\node [style=none] (312) at (-3.75, 1) {};
		\node [style=none] (313) at (-2.25, 1) {};
		\node [style=none] (314) at (-3, 5.5) {};
		\node [style=none] (315) at (-3, 1) {};
		\node [style=none] (316) at (0, 3.5) {};
		\node [style=none] (317) at (0, 1) {};
		\node [style=none] (318) at (4, 1.5) {};
		\node [style=none] (319) at (4, 1) {};
		\node [style=none] (320) at (6.5, 0.25) {};
		\node [style=none] (321) at (12, 2.25) {};
		\node [style=none] (322) at (12, 4.25) {};
		\node [style=none] (323) at (12, 6.25) {};
		\node [style=none] (325) at (-8, 1) {};
		\node [style=none] (326) at (-6.5, 1) {};
		\node [style=none] (327) at (-8, -0.5) {};
		\node [style=none] (328) at (-6.5, -0.5) {};
		\node [style=none] (329) at (-8.5, 0.25) {};
		\node [font={\Large}, style=none] (330) at (-7.25, 0.25) {$Z$};
		\node [style=none] (331) at (-8, 0.25) {};
		\node [style=none] (332) at (8, 0.25) {};
		\node [style=none] (333) at (8.5, 0.25) {};
		\node [style=none] (334) at (8.5, -0.5) {};
		\node [style=none] (335) at (8.5, 1) {};
		\node [style=none] (336) at (10, -0.5) {};
		\node [style=none] (337) at (10, 1) {};
		\node [font={\Large}, style=none] (338) at (9.25, 0.25) {H};
		\node [style=none] (339) at (10.5, 0.25) {};
		\node [style=none] (340) at (10, 0.25) {};
		\node [style=none] (341) at (12, 0.25) {};
		\node [style=none] (342) at (6.5, 1) {};
		\node [style=none] (343) at (8, 1) {};
		\node [style=none] (344) at (6.5, -0.5) {};
		\node [style=none] (345) at (8, -0.5) {};
		\node [font={\Large}, style=none] (347) at (7.25, 0.25) {$Z$};
		\node [style=none] (348) at (10.5, 1) {};
		\node [style=none] (349) at (12, 1) {};
		\node [style=none] (350) at (12, -0.5) {};
		\node [style=none] (351) at (10.5, -0.5) {};
		\node [style=none] (352) at (10.625, 0) {};
		\node [style=none] (353) at (11.875, 0) {};
		\node [style=none] (354) at (11.25, 0) {};
		\node [style=none] (355) at (11.75, 0.75) {};
		\node [style=none] (356) at (1.5, 3.25) {$\vdots$};
	\end{pgfonlayer}
	\begin{pgfonlayer}{edgelayer}
		\draw (221.center) to (217.center);
		\draw (222.center) to (218.center);
		\draw (223.center) to (220.center);
		\draw [style=surface X] (233.center)
			 to (234.center)
			 to (232.center)
			 to (231.center)
			 to cycle;
		\draw [style=surface X] (245.center)
			 to (246.center)
			 to (243.center)
			 to (244.center)
			 to cycle;
		\draw [style=surface X] (239.center)
			 to (240.center)
			 to (242.center)
			 to (241.center)
			 to cycle;
		\draw (247.center) to (226.center);
		\draw (248.center) to (227.center);
		\draw (249.center) to (228.center);
		\draw (250.center) to (225.center);
		\draw (254) to (257);
		\draw (255) to (258);
		\draw (256) to (259);
		\draw [style=surface X] (269.center)
			 to (266.center)
			 to (267.center)
			 to (268.center)
			 to cycle;
		\draw [style=surface X] (275.center)
			 to (276.center)
			 to (274.center)
			 to (273.center)
			 to cycle;
		\draw [style=surface X] (280.center)
			 to (282.center)
			 to (283.center)
			 to (281.center)
			 to cycle;
		\draw [style=surface X] (298.center)
			 to (299.center)
			 to (297.center)
			 to (296.center)
			 to cycle;
		\draw [style=surface X] (312.center)
			 to (313.center)
			 to (310.center)
			 to (311.center)
			 to cycle;
		\draw [style=surface X] (306.center)
			 to (307.center)
			 to (309.center)
			 to (308.center)
			 to cycle;
		\draw (315.center) to (314.center);
		\draw (317.center) to (316.center);
		\draw (319.center) to (318.center);
		\draw (263.center) to (304.center);
		\draw (264.center) to (305.center);
		\draw (320.center) to (265.center);
		\draw (321.center) to (262.center);
		\draw (322.center) to (261.center);
		\draw (323.center) to (260.center);
		\draw [style=surface X] (325.center)
			 to (327.center)
			 to (328.center)
			 to (326.center)
			 to cycle;
		\draw (331.center) to (329.center);
		\draw (333.center) to (332.center);
		\draw (335.center) to (334.center);
		\draw (334.center) to (336.center);
		\draw (336.center) to (337.center);
		\draw (337.center) to (335.center);
		\draw (340.center) to (339.center);
		\draw [style=surface X] (342.center)
			 to (344.center)
			 to (345.center)
			 to (343.center)
			 to cycle;
		\draw (348.center) to (349.center);
		\draw (348.center) to (351.center);
		\draw (351.center) to (350.center);
		\draw (350.center) to (349.center);
		\draw [bend left=90, looseness=1.75] (352.center) to (353.center);
		\draw (355.center) to (354.center);
	\end{pgfonlayer}
\end{tikzpicture}

%% file: figures/Tikz/knill_noise_model.tikz
\begin{tikzpicture}
	\begin{pgfonlayer}{nodelayer}
		\node [style=none] (1) at (-16.5, 5) {};
		\node [style=none] (2) at (-13.5, 3) {};
		\node [style=none] (3) at (-10.5, 1) {};
		\node [style=none] (11) at (8.5, -9) {};
		\node [style=none] (13) at (8.5, -11) {};
		\node [style=none] (15) at (8.5, -13) {};
		\node [style=cnot ctrl] (16) at (-16.5, 5) {};
		\node [style=cnot ctrl] (17) at (-13.5, 3) {};
		\node [style=cnot ctrl] (18) at (-10.5, 1) {};
		\node [style=cnot targ] (19) at (-16.5, -3) {};
		\node [style=cnot targ] (20) at (-13.5, -5) {};
		\node [style=cnot targ] (21) at (-10.5, -7) {};
		\node [style=none] (22) at (-14.25, 5) {};
		\node [style=none] (23) at (-11.25, 3) {};
		\node [style=none] (24) at (-8.25, 1) {};
		\node [style=none] (28) at (-14.25, -3) {};
		\node [style=none] (29) at (-11.25, -5) {};
		\node [style=none] (30) at (-8.25, -7) {};
		\node [font={\Large}, style=none] (34) at (-13.5, 8) {2-qubit Depolarising Channel};
		\node [style=none] (35) at (-7.5, 5.75) {};
		\node [style=none] (36) at (-17.5, 5.75) {};
		\node [style=none] (37) at (-13.5, 7.25) {};
		\node [style=none] (77) at (0.5, 5) {};
		\node [style=none] (78) at (0.5, 3) {};
		\node [style=none] (79) at (0.5, 1) {};
		\node [style=none] (83) at (0.5, 5.5) {};
		\node [style=none] (84) at (0.5, -7.5) {};
		\node [style=none] (85) at (3.5, -7.5) {};
		\node [style=none] (86) at (3.5, 5.5) {};
		\node [font={\Huge}, style=none] (87) at (2, -1) {$\mathcal{D}$};
		\node [style=none] (88) at (0.5, -3) {};
		\node [style=none] (89) at (0.5, -5) {};
		\node [style=none] (90) at (0.5, -7) {};
		\node [style=none] (115) at (4.5, 5) {};
		\node [style=none] (116) at (4.5, 3) {};
		\node [style=none] (117) at (4.5, 1) {};
		\node [style=none] (118) at (3.5, 5) {};
		\node [style=none] (119) at (3.5, 3) {};
		\node [style=none] (120) at (3.5, 1) {};
		\node [style=none] (121) at (4.5, 5.5) {};
		\node [style=none] (122) at (4.5, -7.5) {};
		\node [style=none] (123) at (7.5, -7.5) {};
		\node [style=none] (124) at (7.5, 5.5) {};
		\node [style=none] (126) at (8.5, 5) {};
		\node [style=none] (127) at (8.5, 3) {};
		\node [style=none] (128) at (8.5, 1) {};
		\node [style=none] (129) at (7.5, 5) {};
		\node [style=none] (130) at (7.5, 3) {};
		\node [style=none] (131) at (7.5, 1) {};
		\node [font={\Huge}, style=none] (137) at (6, -1) {$\mathcal{R}$};
		\node [style=none] (139) at (4.5, -3) {};
		\node [style=none] (140) at (4.5, -5) {};
		\node [style=none] (141) at (4.5, -7) {};
		\node [style=none] (142) at (3.5, -3) {};
		\node [style=none] (143) at (3.5, -5) {};
		\node [style=none] (144) at (3.5, -7) {};
		\node [style=none] (150) at (8.5, -3) {};
		\node [style=none] (151) at (8.5, -5) {};
		\node [style=none] (152) at (8.5, -7) {};
		\node [style=none] (153) at (7.5, -3) {};
		\node [style=none] (154) at (7.5, -5) {};
		\node [style=none] (155) at (7.5, -7) {};
		\node [style=none] (169) at (27.25, -9) {};
		\node [style=none] (170) at (27.25, -11) {};
		\node [style=none] (171) at (27.25, -13) {};
		\node [style=none] (172) at (27.25, -9) {};
		\node [style=none] (173) at (26, -9) {};
		\node [style=none] (174) at (26, -11) {};
		\node [style=none] (175) at (26, -13) {};
		\node [style=none] (176) at (27.25, -11) {};
		\node [style=none] (268) at (-15.75, 5.75) {};
		\node [style=none] (269) at (-15.75, 4.25) {};
		\node [style=none] (270) at (-14.25, 4.25) {};
		\node [style=none] (271) at (-14.25, 5.75) {};
		\node [font={\Large}, style=none] (272) at (-15, 5) {$D_{2}$};
		\node [style=none] (273) at (-16.5, 5) {};
		\node [style=none] (274) at (-16.5, 5) {};
		\node [style=none] (275) at (-12.75, 3.75) {};
		\node [style=none] (276) at (-11.25, 3.75) {};
		\node [style=none] (277) at (-12.75, 2.25) {};
		\node [style=none] (278) at (-11.25, 2.25) {};
		\node [style=none] (279) at (-13.5, 3) {};
		\node [style=none] (280) at (-13.5, 3) {};
		\node [font={\Large}, style=none] (281) at (-12, 3) {$D_{2}$};
		\node [style=none] (284) at (-9.75, 1.75) {};
		\node [style=none] (285) at (-8.25, 1.75) {};
		\node [style=none] (286) at (-9.75, 0.25) {};
		\node [style=none] (287) at (-8.25, 0.25) {};
		\node [style=none] (288) at (-10.5, 1) {};
		\node [style=none] (289) at (-10.5, 1) {};
		\node [font={\Large}, style=none] (290) at (-9, 1) {$D_{2}$};
		\node [style=none] (291) at (-15.75, 5) {};
		\node [style=none] (292) at (-12.75, 3) {};
		\node [style=none] (293) at (-9.75, 1) {};
		\node [font={\Large}, style=none] (296) at (-15, -3) {$D_{2}$};
		\node [style=none] (297) at (-16.5, -3) {};
		\node [style=none] (298) at (-16.5, -3) {};
		\node [style=none] (299) at (-13.5, -5) {};
		\node [style=none] (300) at (-13.5, -5) {};
		\node [font={\Large}, style=none] (301) at (-12, -5) {$D_{2}$};
		\node [style=none] (303) at (-9.75, -6.25) {};
		\node [style=none] (304) at (-8.25, -6.25) {};
		\node [style=none] (305) at (-9.75, -7.75) {};
		\node [style=none] (306) at (-8.25, -7.75) {};
		\node [style=none] (307) at (-10.5, -7) {};
		\node [style=none] (308) at (-10.5, -7) {};
		\node [font={\Large}, style=none] (309) at (-9, -7) {$D_{2}$};
		\node [style=none] (310) at (-15.75, -3) {};
		\node [style=none] (311) at (-12.75, -5) {};
		\node [style=none] (312) at (-9.75, -7) {};
		\node [style=none] (316) at (-18.5, -9) {};
		\node [style=none] (318) at (-18.5, -11) {};
		\node [style=none] (321) at (-18.5, -13) {};
		\node [style=none] (322) at (8.5, -9) {};
		\node [style=none] (323) at (8.5, -11) {};
		\node [style=none] (324) at (8.5, -13) {};
		\node [style=none] (329) at (-12.75, -4.25) {};
		\node [style=none] (330) at (-11.25, -4.25) {};
		\node [style=none] (331) at (-12.75, -5.75) {};
		\node [style=none] (332) at (-11.25, -5.75) {};
		\node [style=none] (333) at (-14.25, -3.75) {};
		\node [style=none] (334) at (-15.75, -3.75) {};
		\node [style=none] (335) at (-15.75, -2.25) {};
		\node [style=none] (336) at (-14.25, -2.25) {};
		\node [style=none] (337) at (-15, 4.25) {};
		\node [style=none] (338) at (-15, -2.25) {};
		\node [style=none] (339) at (-12, 2.25) {};
		\node [style=none] (340) at (-12, -4.25) {};
		\node [style=none] (341) at (-9, 0.25) {};
		\node [style=none] (342) at (-9, -6.25) {};
		\node [style=none] (454) at (-18.5, 5) {};
		\node [style=none] (460) at (-18.5, 3) {};
		\node [style=none] (469) at (-18.5, 1) {};
		\node [style=none] (475) at (-18.5, -3) {};
		\node [style=none] (477) at (-18.5, -5) {};
		\node [style=none] (485) at (-18.5, -7) {};
		\node [style=none] (533) at (-4.5, 3) {};
		\node [style=none] (534) at (-4.5, 1) {};
		\node [style=none] (535) at (-4, 5) {};
		\node [style=none] (536) at (-4, 3) {};
		\node [style=none] (537) at (-4, 1) {};
		\node [font={\Large}, style=none] (538) at (-3.5, 7.5) {$Z$-error channel};
		\node [style=none] (539) at (-0.25, 5.75) {};
		\node [style=none] (540) at (-6.5, 5.75) {};
		\node [style=none] (541) at (-3.5, 7) {};
		\node [style=none] (542) at (-4, 4.25) {};
		\node [style=none] (543) at (-4, 5.75) {};
		\node [style=none] (544) at (-2.5, 4.25) {};
		\node [style=none] (545) at (-2.5, 5.75) {};
		\node [font={\Large}, style=none] (546) at (-3.25, 5) {H};
		\node [style=none] (547) at (-4, 2.25) {};
		\node [style=none] (548) at (-4, 3.75) {};
		\node [style=none] (549) at (-2.5, 2.25) {};
		\node [style=none] (550) at (-2.5, 3.75) {};
		\node [font={\Large}, style=none] (551) at (-3.25, 3) {H};
		\node [style=none] (552) at (-4, 0.25) {};
		\node [style=none] (553) at (-4, 1.75) {};
		\node [style=none] (554) at (-2.5, 0.25) {};
		\node [style=none] (555) at (-2.5, 1.75) {};
		\node [font={\Large}, style=none] (556) at (-3.25, 1) {H};
		\node [style=none] (559) at (-2, 5) {};
		\node [style=none] (560) at (-2.5, 5) {};
		\node [style=none] (563) at (-2, 3) {};
		\node [style=none] (564) at (-2.5, 3) {};
		\node [style=none] (567) at (-2, 1) {};
		\node [style=none] (568) at (-2.5, 1) {};
		\node [style=none] (569) at (-0.5, 5) {};
		\node [style=none] (570) at (-0.5, 3) {};
		\node [style=none] (571) at (-0.5, 1) {};
		\node [style=none] (572) at (-6, 5.75) {};
		\node [style=none] (573) at (-6, 4.25) {};
		\node [style=none] (574) at (-4.5, 4.25) {};
		\node [style=none] (575) at (-4.5, 5.75) {};
		\node [font={\Large}, style=none] (576) at (-5.25, 5) {$Z$};
		\node [style=none] (577) at (-4.5, 5) {};
		\node [style=none] (578) at (-6, 3.75) {};
		\node [style=none] (579) at (-4.5, 3.75) {};
		\node [style=none] (580) at (-6, 2.25) {};
		\node [style=none] (581) at (-4.5, 2.25) {};
		\node [font={\Large}, style=none] (582) at (-5.25, 3) {$Z$};
		\node [style=none] (583) at (-6, 1.75) {};
		\node [style=none] (584) at (-4.5, 1.75) {};
		\node [style=none] (585) at (-6, 0.25) {};
		\node [style=none] (586) at (-4.5, 0.25) {};
		\node [style=none] (587) at (-6, 1) {};
		\node [font={\Large}, style=none] (588) at (-5.25, 1) {$Z$};
		\node [style=none] (589) at (-6, 3) {};
		\node [style=none] (594) at (-6, 5) {};
		\node [style=none] (595) at (-3.5, -3) {};
		\node [style=none] (596) at (-3.5, -5) {};
		\node [style=none] (597) at (-3.5, -7) {};
		\node [style=none] (598) at (-2.75, -3) {};
		\node [style=none] (599) at (-2.75, -5) {};
		\node [style=none] (600) at (-2.75, -7) {};
		\node [font={\Large}, style=none] (601) at (-3.5, -0.75) {X-error Channel};
		\node [style=none] (602) at (-0.25, -2.25) {};
		\node [style=none] (603) at (-5.5, -2.25) {};
		\node [style=none] (604) at (-3, -1.5) {};
		\node [style=none] (607) at (-2, -3) {};
		\node [style=none] (608) at (-2.75, -3) {};
		\node [style=none] (611) at (-2, -5) {};
		\node [style=none] (612) at (-2.75, -5) {};
		\node [style=none] (615) at (-2, -7) {};
		\node [style=none] (616) at (-2.75, -7) {};
		\node [style=none] (617) at (-0.5, -3) {};
		\node [style=none] (618) at (-0.5, -7) {};
		\node [style=none] (619) at (-0.5, -5) {};
		\node [font={\Large}, style=none] (620) at (-4.25, -3) {$X$};
		\node [font={\Large}, style=none] (621) at (-4.25, -5) {$X$};
		\node [style=none] (622) at (-5, -6.25) {};
		\node [style=none] (623) at (-3.5, -6.25) {};
		\node [style=none] (624) at (-5, -7.75) {};
		\node [style=none] (625) at (-3.5, -7.75) {};
		\node [font={\Large}, style=none] (626) at (-4.25, -7) {$X$};
		\node [style=none] (627) at (-5, -3) {};
		\node [style=none] (628) at (-5, -5) {};
		\node [style=none] (629) at (-5, -7) {};
		\node [style=none] (634) at (-5, -4.25) {};
		\node [style=none] (635) at (-3.5, -4.25) {};
		\node [style=none] (636) at (-5, -5.75) {};
		\node [style=none] (637) at (-3.5, -5.75) {};
		\node [style=none] (638) at (-3.5, -3.75) {};
		\node [style=none] (639) at (-5, -3.75) {};
		\node [style=none] (640) at (-5, -2.25) {};
		\node [style=none] (641) at (-3.5, -2.25) {};
		\node [style=none] (643) at (-4.25, -4.5) {};
		\node [style=none] (644) at (-4.25, -6.5) {};
		\node [style=none] (646) at (-2, 5.75) {};
		\node [style=none] (647) at (-0.5, 5.75) {};
		\node [style=none] (648) at (-0.5, 4.25) {};
		\node [style=none] (649) at (-2, 4.25) {};
		\node [style=none] (650) at (-1.875, 4.75) {};
		\node [style=none] (651) at (-0.625, 4.75) {};
		\node [style=none] (652) at (-1.25, 4.75) {};
		\node [style=none] (653) at (-0.75, 5.5) {};
		\node [style=none] (655) at (-2, 3.75) {};
		\node [style=none] (656) at (-0.5, 3.75) {};
		\node [style=none] (657) at (-0.5, 2.25) {};
		\node [style=none] (658) at (-2, 2.25) {};
		\node [style=none] (659) at (-1.875, 2.75) {};
		\node [style=none] (660) at (-0.625, 2.75) {};
		\node [style=none] (661) at (-1.25, 2.75) {};
		\node [style=none] (662) at (-0.75, 3.5) {};
		\node [style=none] (664) at (-2, 1.75) {};
		\node [style=none] (665) at (-0.5, 1.75) {};
		\node [style=none] (666) at (-0.5, 0.25) {};
		\node [style=none] (667) at (-2, 0.25) {};
		\node [style=none] (668) at (-1.875, 0.75) {};
		\node [style=none] (669) at (-0.625, 0.75) {};
		\node [style=none] (670) at (-1.25, 0.75) {};
		\node [style=none] (671) at (-0.75, 1.5) {};
		\node [style=none] (673) at (-2, -2.25) {};
		\node [style=none] (674) at (-0.5, -2.25) {};
		\node [style=none] (675) at (-0.5, -3.75) {};
		\node [style=none] (676) at (-2, -3.75) {};
		\node [style=none] (677) at (-1.875, -3.25) {};
		\node [style=none] (678) at (-0.625, -3.25) {};
		\node [style=none] (679) at (-1.25, -3.25) {};
		\node [style=none] (680) at (-0.75, -2.5) {};
		\node [style=none] (682) at (-2, -4.25) {};
		\node [style=none] (683) at (-0.5, -4.25) {};
		\node [style=none] (684) at (-0.5, -5.75) {};
		\node [style=none] (685) at (-2, -5.75) {};
		\node [style=none] (686) at (-1.875, -5.25) {};
		\node [style=none] (687) at (-0.625, -5.25) {};
		\node [style=none] (688) at (-1.25, -5.25) {};
		\node [style=none] (689) at (-0.75, -4.5) {};
		\node [style=none] (691) at (-2, -6.25) {};
		\node [style=none] (692) at (-0.5, -6.25) {};
		\node [style=none] (693) at (-0.5, -7.75) {};
		\node [style=none] (694) at (-2, -7.75) {};
		\node [style=none] (695) at (-1.875, -7.25) {};
		\node [style=none] (696) at (-0.625, -7.25) {};
		\node [style=none] (697) at (-1.25, -7.25) {};
		\node [style=none] (698) at (-0.75, -6.5) {};
		\node [style=none] (702) at (8.5, -8.5) {};
		\node [style=none] (703) at (8.5, -13.5) {};
		\node [style=none] (704) at (11.5, -13.5) {};
		\node [style=none] (705) at (11.5, -8.5) {};
		\node [font={\Huge}, style=none] (706) at (10, -11) {$\overline{X_1}$};
		\node [style=none] (707) at (12.5, -9) {};
		\node [style=none] (708) at (12.5, -11) {};
		\node [style=none] (709) at (12.5, -13) {};
		\node [style=none] (710) at (12.5, -8.5) {};
		\node [style=none] (711) at (12.5, -13.5) {};
		\node [style=none] (712) at (15.5, -13.5) {};
		\node [style=none] (713) at (15.5, -8.5) {};
		\node [font={\Huge}, style=none] (714) at (14, -11) {$\overline{Z_1}$};
		\node [style=none] (715) at (11.5, -9) {};
		\node [style=none] (716) at (11.5, -11) {};
		\node [style=none] (717) at (11.5, -13) {};
		\node [style=none] (724) at (8.5, 5.5) {};
		\node [style=none] (725) at (8.5, -7.5) {};
		\node [style=none] (726) at (11.5, -7.5) {};
		\node [style=none] (727) at (11.5, 5.5) {};
		\node [style=none] (730) at (10, -8.5) {};
		\node [style=none] (731) at (10, -7.5) {};
		\node [style=none] (732) at (14, -8.5) {};
		\node [style=none] (733) at (14, -7.5) {};
		\node [style=none] (737) at (12.5, 5) {};
		\node [style=none] (738) at (12.5, 3) {};
		\node [style=none] (739) at (12.5, 1) {};
		\node [style=none] (740) at (11.5, 5) {};
		\node [style=none] (741) at (11.5, 3) {};
		\node [style=none] (742) at (11.5, 1) {};
		\node [style=none] (743) at (12.5, -3) {};
		\node [style=none] (744) at (12.5, -5) {};
		\node [style=none] (745) at (12.5, -7) {};
		\node [style=none] (746) at (11.5, -3) {};
		\node [style=none] (747) at (11.5, -5) {};
		\node [style=none] (748) at (11.5, -7) {};
		\node [style=none] (749) at (12.5, 5.5) {};
		\node [style=none] (750) at (12.5, -7.5) {};
		\node [style=none] (751) at (15.5, -7.5) {};
		\node [style=none] (752) at (15.5, 5.5) {};
		\node [style=none] (755) at (19, -8.5) {};
		\node [style=none] (756) at (19, -13.5) {};
		\node [style=none] (757) at (22, -13.5) {};
		\node [style=none] (758) at (22, -8.5) {};
		\node [font={\Huge}, style=none] (759) at (20.5, -11) {$\overline{X_k}$};
		\node [style=none] (760) at (23, -9) {};
		\node [style=none] (761) at (23, -11) {};
		\node [style=none] (762) at (23, -13) {};
		\node [style=none] (763) at (23, -8.5) {};
		\node [style=none] (764) at (23, -13.5) {};
		\node [style=none] (765) at (26, -13.5) {};
		\node [style=none] (766) at (26, -8.5) {};
		\node [font={\Huge}, style=none] (767) at (24.5, -11) {$\overline{Z_k}$};
		\node [style=none] (768) at (22, -9) {};
		\node [style=none] (769) at (22, -11) {};
		\node [style=none] (770) at (22, -13) {};
		\node [style=none] (771) at (19, 5) {};
		\node [style=none] (772) at (19, 3) {};
		\node [style=none] (773) at (19, 1) {};
		\node [style=none] (774) at (18, 5) {};
		\node [style=none] (775) at (18, 3) {};
		\node [style=none] (776) at (18, 1) {};
		\node [style=none] (777) at (19, -3) {};
		\node [style=none] (778) at (19, -5) {};
		\node [style=none] (779) at (19, -7) {};
		\node [style=none] (780) at (18, -3) {};
		\node [style=none] (781) at (18, -5) {};
		\node [style=none] (782) at (18, -7) {};
		\node [style=none] (783) at (19, 5.5) {};
		\node [style=none] (784) at (19, -7.5) {};
		\node [style=none] (785) at (22, -7.5) {};
		\node [style=none] (786) at (22, 5.5) {};
		\node [style=none] (789) at (20.5, -8.5) {};
		\node [style=none] (790) at (20.5, -7.5) {};
		\node [style=none] (791) at (24.5, -8.5) {};
		\node [style=none] (792) at (24.5, -7.5) {};
		\node [style=none] (793) at (23, 5) {};
		\node [style=none] (794) at (23, 3) {};
		\node [style=none] (795) at (23, 1) {};
		\node [style=none] (796) at (22, 5) {};
		\node [style=none] (797) at (22, 3) {};
		\node [style=none] (798) at (22, 1) {};
		\node [style=none] (799) at (23, -3) {};
		\node [style=none] (800) at (23, -5) {};
		\node [style=none] (801) at (23, -7) {};
		\node [style=none] (802) at (22, -3) {};
		\node [style=none] (803) at (22, -5) {};
		\node [style=none] (804) at (22, -7) {};
		\node [style=none] (805) at (23, 5.5) {};
		\node [style=none] (806) at (23, -7.5) {};
		\node [style=none] (807) at (26, -7.5) {};
		\node [style=none] (808) at (26, 5.5) {};
		\node [style=none] (811) at (16.5, 5) {};
		\node [style=none] (812) at (16.5, 3) {};
		\node [style=none] (813) at (16.5, 1) {};
		\node [style=none] (814) at (15.5, 5) {};
		\node [style=none] (815) at (15.5, 3) {};
		\node [style=none] (816) at (15.5, 1) {};
		\node [style=none] (817) at (16.5, -3) {};
		\node [style=none] (818) at (16.5, -5) {};
		\node [style=none] (819) at (16.5, -7) {};
		\node [style=none] (820) at (15.5, -3) {};
		\node [style=none] (821) at (15.5, -5) {};
		\node [style=none] (822) at (15.5, -7) {};
		\node [style=none] (823) at (17.25, -0.5) {$\cdots$};
		\node [style=none] (824) at (16.5, -9) {};
		\node [style=none] (825) at (16.5, -11) {};
		\node [style=none] (826) at (16.5, -13) {};
		\node [style=none] (827) at (15.5, -9) {};
		\node [style=none] (828) at (15.5, -11) {};
		\node [style=none] (829) at (15.5, -13) {};
		\node [style=none] (830) at (19, -9) {};
		\node [style=none] (831) at (19, -11) {};
		\node [style=none] (832) at (19, -13) {};
		\node [style=none] (833) at (18, -9) {};
		\node [style=none] (834) at (18, -11) {};
		\node [style=none] (835) at (18, -13) {};
		\node [style=none] (836) at (17.25, -11) {$\cdots$};
		\node [style=none] (844) at (-28.5, -8) {\Huge{$\left(\frac{1}{\sqrt{2}}\left[\ket{\overline{00}}+\ket{\overline{11}}\right]\right)^{\otimes k}$}};
		\node [style=none] (846) at (-18.5, -2.5) {};
		\node [style=none] (847) at (-18.5, -13.5) {};
		\node [style=none] (848) at (-20.5, -8) {};
		\node [font={\Huge}, style=none] (849) at (10, 3) {$\overline{Z_1}$};
		\node [style=none] (850) at (10, 4.75) {Calc.};
		\node [font={\Huge}, style=none] (851) at (10, -5) {$\overline{Z_{1}}$};
		\node [style=none] (852) at (10, -1) {\Large $\otimes$};
		\node [font={\Huge}, style=none] (853) at (14, 3) {$\overline{X_1}$};
		\node [style=none] (854) at (14, 4.75) {Calc.};
		\node [font={\Huge}, style=none] (855) at (14, -5) {$\overline{X_{1}}$};
		\node [style=none] (856) at (14, -1) {\Large $\otimes$};
		\node [font={\Huge}, style=none] (857) at (20.5, 3) {$\overline{Z_k}$};
		\node [style=none] (858) at (20.5, 4.75) {Calc.};
		\node [font={\Huge}, style=none] (859) at (20.5, -5) {$\overline{Z_k}$};
		\node [style=none] (860) at (20.5, -1) {\Large $\otimes$};
		\node [font={\Huge}, style=none] (861) at (24.5, 3) {$\overline{X_k}$};
		\node [style=none] (862) at (24.5, 4.75) {Calc.};
		\node [font={\Huge}, style=none] (863) at (24.5, -5) {$\overline{X_k}$};
		\node [style=none] (864) at (24.5, -1) {\Large $\otimes$};
	\end{pgfonlayer}
	\begin{pgfonlayer}{edgelayer}
		\draw (16) to (19);
		\draw (17) to (20);
		\draw (18) to (21);
		\draw [in=90, out=-90, looseness=0.75] (37.center) to (35.center);
		\draw [in=-90, out=90] (36.center) to (37.center);
		\draw (84.center)
			 to (85.center)
			 to (86.center)
			 to (83.center)
			 to cycle;
		\draw [style=double] (118.center) to (115.center);
		\draw [style=double] (119.center) to (116.center);
		\draw [style=double] (120.center) to (117.center);
		\draw [style=none] (124.center)
			 to (121.center)
			 to (122.center)
			 to (123.center)
			 to cycle;
		\draw [style=double] (129.center) to (126.center);
		\draw [style=double] (130.center) to (127.center);
		\draw [style=double] (131.center) to (128.center);
		\draw [style=double] (142.center) to (139.center);
		\draw [style=double] (143.center) to (140.center);
		\draw [style=double] (144.center) to (141.center);
		\draw [style=double] (153.center) to (150.center);
		\draw [style=double] (154.center) to (151.center);
		\draw [style=double] (155.center) to (152.center);
		\draw [style=not synthesisable] (174.center) to (176.center);
		\draw [style=not synthesisable] (175.center) to (171.center);
		\draw (173.center) to (172.center);
		\draw (291.center) to (274.center);
		\draw (292.center) to (280.center);
		\draw (293.center) to (289.center);
		\draw (310.center) to (298.center);
		\draw (311.center) to (300.center);
		\draw (312.center) to (308.center);
		\draw (322.center) to (316.center);
		\draw (323.center) to (318.center);
		\draw (324.center) to (321.center);
		\draw [style=surface X] (271.center)
			 to (268.center)
			 to (269.center)
			 to (270.center)
			 to cycle;
		\draw [style=surface X] (277.center)
			 to (278.center)
			 to (276.center)
			 to (275.center)
			 to cycle;
		\draw [style=surface X] (284.center)
			 to (286.center)
			 to (287.center)
			 to (285.center)
			 to cycle;
		\draw [style=surface X] (305.center)
			 to (306.center)
			 to (304.center)
			 to (303.center)
			 to cycle;
		\draw [style=surface X] (335.center)
			 to (336.center)
			 to (333.center)
			 to (334.center)
			 to cycle;
		\draw [style=surface X] (329.center)
			 to (330.center)
			 to (332.center)
			 to (331.center)
			 to cycle;
		\draw (338.center) to (337.center);
		\draw (340.center) to (339.center);
		\draw (342.center) to (341.center);
		\draw (454.center) to (274.center);
		\draw (460.center) to (280.center);
		\draw (469.center) to (289.center);
		\draw (475.center) to (298.center);
		\draw (477.center) to (300.center);
		\draw (485.center) to (308.center);
		\draw (536.center) to (533.center);
		\draw (537.center) to (534.center);
		\draw [in=90, out=-90] (541.center) to (539.center);
		\draw [in=-90, out=90] (540.center) to (541.center);
		\draw (543.center) to (542.center);
		\draw (542.center) to (544.center);
		\draw (544.center) to (545.center);
		\draw (545.center) to (543.center);
		\draw (548.center) to (547.center);
		\draw (547.center) to (549.center);
		\draw (549.center) to (550.center);
		\draw (550.center) to (548.center);
		\draw (553.center) to (552.center);
		\draw (552.center) to (554.center);
		\draw (554.center) to (555.center);
		\draw (555.center) to (553.center);
		\draw (560.center) to (559.center);
		\draw (564.center) to (563.center);
		\draw (568.center) to (567.center);
		\draw [style=surface X] (575.center)
			 to (572.center)
			 to (573.center)
			 to (574.center)
			 to cycle;
		\draw [style=surface X] (580.center)
			 to (581.center)
			 to (579.center)
			 to (578.center)
			 to cycle;
		\draw [style=surface X] (583.center)
			 to (585.center)
			 to (586.center)
			 to (584.center)
			 to cycle;
		\draw (22.center) to (594.center);
		\draw (23.center) to (589.center);
		\draw (24.center) to (587.center);
		\draw (598.center) to (595.center);
		\draw (599.center) to (596.center);
		\draw (600.center) to (597.center);
		\draw [in=90, out=-90] (604.center) to (602.center);
		\draw [in=-90, out=90] (603.center) to (604.center);
		\draw (608.center) to (607.center);
		\draw (612.center) to (611.center);
		\draw (616.center) to (615.center);
		\draw [style=surface X] (624.center)
			 to (625.center)
			 to (623.center)
			 to (622.center)
			 to cycle;
		\draw [style=surface X] (640.center)
			 to (641.center)
			 to (638.center)
			 to (639.center)
			 to cycle;
		\draw [style=surface X] (634.center)
			 to (635.center)
			 to (637.center)
			 to (636.center)
			 to cycle;
		\draw (28.center) to (627.center);
		\draw (29.center) to (628.center);
		\draw (30.center) to (629.center);
		\draw (646.center) to (647.center);
		\draw (646.center) to (649.center);
		\draw (649.center) to (648.center);
		\draw (648.center) to (647.center);
		\draw [bend left=90, looseness=1.75] (650.center) to (651.center);
		\draw (653.center) to (652.center);
		\draw (655.center) to (656.center);
		\draw (655.center) to (658.center);
		\draw (658.center) to (657.center);
		\draw (657.center) to (656.center);
		\draw [bend left=90, looseness=1.75] (659.center) to (660.center);
		\draw (662.center) to (661.center);
		\draw (664.center) to (665.center);
		\draw (664.center) to (667.center);
		\draw (667.center) to (666.center);
		\draw (666.center) to (665.center);
		\draw [bend left=90, looseness=1.75] (668.center) to (669.center);
		\draw (671.center) to (670.center);
		\draw (673.center) to (674.center);
		\draw (673.center) to (676.center);
		\draw (676.center) to (675.center);
		\draw (675.center) to (674.center);
		\draw [bend left=90, looseness=1.75] (677.center) to (678.center);
		\draw (680.center) to (679.center);
		\draw (682.center) to (683.center);
		\draw (682.center) to (685.center);
		\draw (685.center) to (684.center);
		\draw (684.center) to (683.center);
		\draw [bend left=90, looseness=1.75] (686.center) to (687.center);
		\draw (689.center) to (688.center);
		\draw (691.center) to (692.center);
		\draw (691.center) to (694.center);
		\draw (694.center) to (693.center);
		\draw (693.center) to (692.center);
		\draw [bend left=90, looseness=1.75] (695.center) to (696.center);
		\draw (698.center) to (697.center);
		\draw (703.center) to (702.center);
		\draw (702.center) to (705.center);
		\draw (705.center) to (704.center);
		\draw (704.center) to (703.center);
		\draw (711.center) to (710.center);
		\draw (710.center) to (713.center);
		\draw (713.center) to (712.center);
		\draw (712.center) to (711.center);
		\draw (715.center) to (707.center);
		\draw (716.center) to (708.center);
		\draw (717.center) to (709.center);
		\draw [style=none] (724.center)
			 to (725.center)
			 to (726.center)
			 to (727.center)
			 to cycle;
		\draw [style=double] (731.center) to (730.center);
		\draw [style=double] (733.center) to (732.center);
		\draw [style=double] (740.center) to (737.center);
		\draw [style=double] (741.center) to (738.center);
		\draw [style=double] (742.center) to (739.center);
		\draw [style=double] (746.center) to (743.center);
		\draw [style=double] (747.center) to (744.center);
		\draw [style=double] (748.center) to (745.center);
		\draw [style=none] (751.center)
			 to (752.center)
			 to (749.center)
			 to (750.center)
			 to cycle;
		\draw (756.center) to (755.center);
		\draw (755.center) to (758.center);
		\draw (758.center) to (757.center);
		\draw (757.center) to (756.center);
		\draw (764.center) to (763.center);
		\draw (763.center) to (766.center);
		\draw (766.center) to (765.center);
		\draw (765.center) to (764.center);
		\draw (768.center) to (760.center);
		\draw (769.center) to (761.center);
		\draw (770.center) to (762.center);
		\draw [style=double] (774.center) to (771.center);
		\draw [style=double] (775.center) to (772.center);
		\draw [style=double] (776.center) to (773.center);
		\draw [style=double] (780.center) to (777.center);
		\draw [style=double] (781.center) to (778.center);
		\draw [style=double] (782.center) to (779.center);
		\draw [style=none] (786.center)
			 to (783.center)
			 to (784.center)
			 to (785.center)
			 to cycle;
		\draw [style=double] (790.center) to (789.center);
		\draw [style=double] (792.center) to (791.center);
		\draw [style=double] (796.center) to (793.center);
		\draw [style=double] (797.center) to (794.center);
		\draw [style=double] (798.center) to (795.center);
		\draw [style=double] (802.center) to (799.center);
		\draw [style=double] (803.center) to (800.center);
		\draw [style=double] (804.center) to (801.center);
		\draw [style=none] (808.center)
			 to (805.center)
			 to (806.center)
			 to (807.center)
			 to cycle;
		\draw [style=double] (814.center) to (811.center);
		\draw [style=double] (815.center) to (812.center);
		\draw [style=double] (816.center) to (813.center);
		\draw [style=double] (820.center) to (817.center);
		\draw [style=double] (821.center) to (818.center);
		\draw [style=double] (822.center) to (819.center);
		\draw (827.center) to (824.center);
		\draw (828.center) to (825.center);
		\draw (829.center) to (826.center);
		\draw (833.center) to (830.center);
		\draw (834.center) to (831.center);
		\draw (835.center) to (832.center);
		\draw [style=double] (569.center) to (77.center);
		\draw [style=double] (570.center) to (78.center);
		\draw [style=double] (571.center) to (79.center);
		\draw [style=double] (617.center) to (88.center);
		\draw [style=double] (619.center) to (89.center);
		\draw [style=double] (618.center) to (90.center);
		\draw (535.center) to (577.center);
		\draw [in=-180, out=0] (848.center) to (846.center);
		\draw [in=0, out=-180] (847.center) to (848.center);
	\end{pgfonlayer}
\end{tikzpicture}

%% file: figures/Tikz/knill_compressed.tikz
\begin{tikzpicture}
	\begin{pgfonlayer}{nodelayer}
		\node [style=none] (0) at (71.75, 2.5) {\Huge{$\ket{\overline{\psi}}$}};
		\node [style=none] (4) at (20.5, 20.5) {};
		\node [style=none] (5) at (20.5, 18.5) {};
		\node [style=none] (6) at (20.5, 16.5) {};
		\node [style=none] (7) at (20.5, 21.25) {};
		\node [style=none] (8) at (20.5, 15.75) {};
		\node [style=none] (9) at (23.5, 15.75) {};
		\node [style=none] (10) at (23.5, 21.25) {};
		\node [font={\Huge}, style=none] (11) at (22, 18.5) {$\mathcal{D}_X$};
		\node [style=none] (15) at (24.5, 20.5) {};
		\node [style=none] (16) at (24.5, 18.5) {};
		\node [style=none] (17) at (24.5, 16.5) {};
		\node [style=none] (18) at (23.5, 20.5) {};
		\node [style=none] (19) at (23.5, 18.5) {};
		\node [style=none] (20) at (23.5, 16.5) {};
		\node [style=none] (21) at (24.5, 21.25) {};
		\node [style=none] (22) at (24.5, 15.75) {};
		\node [style=none] (23) at (27.5, 15.75) {};
		\node [style=none] (24) at (27.5, 21.25) {};
		\node [style=none] (28) at (27.5, 20.5) {};
		\node [style=none] (29) at (27.5, 18.5) {};
		\node [style=none] (30) at (27.5, 16.5) {};
		\node [font={\Huge}, style=none] (31) at (26, 18.5) {$\mathcal{R}$};
		\node [style=none] (52) at (68.5, 0) {};
		\node [style=none] (53) at (68.5, 5) {};
		\node [style=none] (54) at (69.5, 2.5) {};
		\node [style=none] (58) at (16, 20.5) {};
		\node [style=none] (59) at (16, 18.5) {};
		\node [style=none] (60) at (16, 16.5) {};
		\node [style=none] (61) at (16, 19.75) {};
		\node [style=none] (62) at (16, 21.25) {};
		\node [style=none] (63) at (17.5, 19.75) {};
		\node [style=none] (64) at (17.5, 21.25) {};
		\node [font={\Large}, style=none] (65) at (16.75, 20.5) {H};
		\node [style=none] (66) at (16, 17.75) {};
		\node [style=none] (67) at (16, 19.25) {};
		\node [style=none] (68) at (17.5, 17.75) {};
		\node [style=none] (69) at (17.5, 19.25) {};
		\node [font={\Large}, style=none] (70) at (16.75, 18.5) {H};
		\node [style=none] (71) at (16, 15.75) {};
		\node [style=none] (72) at (16, 17.25) {};
		\node [style=none] (73) at (17.5, 15.75) {};
		\node [style=none] (74) at (17.5, 17.25) {};
		\node [font={\Large}, style=none] (75) at (16.75, 16.5) {H};
		\node [style=none] (76) at (18, 20.5) {};
		\node [style=none] (77) at (17.5, 20.5) {};
		\node [style=none] (78) at (18, 18.5) {};
		\node [style=none] (79) at (17.5, 18.5) {};
		\node [style=none] (80) at (18, 16.5) {};
		\node [style=none] (81) at (17.5, 16.5) {};
		\node [style=none] (82) at (19.5, 20.5) {};
		\node [style=none] (83) at (19.5, 18.5) {};
		\node [style=none] (84) at (19.5, 16.5) {};
		\node [style=none] (97) at (18, 21.25) {};
		\node [style=none] (98) at (19.5, 21.25) {};
		\node [style=none] (99) at (19.5, 19.75) {};
		\node [style=none] (100) at (18, 19.75) {};
		\node [style=none] (101) at (18.125, 20.25) {};
		\node [style=none] (102) at (19.375, 20.25) {};
		\node [style=none] (103) at (18.75, 20.25) {};
		\node [style=none] (104) at (19.25, 21) {};
		\node [style=none] (105) at (18, 19.25) {};
		\node [style=none] (106) at (19.5, 19.25) {};
		\node [style=none] (107) at (19.5, 17.75) {};
		\node [style=none] (108) at (18, 17.75) {};
		\node [style=none] (109) at (18.125, 18.25) {};
		\node [style=none] (110) at (19.375, 18.25) {};
		\node [style=none] (111) at (18.75, 18.25) {};
		\node [style=none] (112) at (19.25, 19) {};
		\node [style=none] (113) at (18, 17.25) {};
		\node [style=none] (114) at (19.5, 17.25) {};
		\node [style=none] (115) at (19.5, 15.75) {};
		\node [style=none] (116) at (18, 15.75) {};
		\node [style=none] (117) at (18.125, 16.25) {};
		\node [style=none] (118) at (19.375, 16.25) {};
		\node [style=none] (119) at (18.75, 16.25) {};
		\node [style=none] (120) at (19.25, 17) {};
		\node [style=none] (150) at (28.5, 12.5) {};
		\node [style=none] (151) at (28.5, 10.5) {};
		\node [style=none] (152) at (28.5, 8.5) {};
		\node [style=none] (153) at (28.5, 13.25) {};
		\node [style=none] (154) at (28.5, 7.75) {};
		\node [style=none] (155) at (31.5, 7.75) {};
		\node [style=none] (156) at (31.5, 13.25) {};
		\node [font={\Huge}, style=none] (157) at (30, 10.5) {$\overline{Z_1}$};
		\node [style=none] (168) at (30, 13.25) {};
		\node [style=none] (169) at (30, 15.75) {};
		\node [style=none] (170) at (28.5, 20.5) {};
		\node [style=none] (171) at (28.5, 18.5) {};
		\node [style=none] (172) at (28.5, 16.5) {};
		\node [style=none] (182) at (28.5, 21.25) {};
		\node [style=none] (183) at (28.5, 15.75) {};
		\node [style=none] (184) at (31.5, 15.75) {};
		\node [style=none] (185) at (31.5, 21.25) {};
		\node [font={\Huge}, style=none] (186) at (30, 17.5) {$\overline{X_1}$};
		\node [style=none] (192) at (35, 12.5) {};
		\node [style=none] (193) at (35, 10.5) {};
		\node [style=none] (194) at (35, 8.5) {};
		\node [style=none] (195) at (35, 13.25) {};
		\node [style=none] (196) at (35, 7.75) {};
		\node [style=none] (197) at (38, 7.75) {};
		\node [style=none] (198) at (38, 13.25) {};
		\node [font={\Huge}, style=none] (199) at (36.5, 10.5) {$\overline{Z_k}$};
		\node [style=none] (206) at (34, 20.5) {};
		\node [style=none] (207) at (34, 18.5) {};
		\node [style=none] (208) at (34, 16.5) {};
		\node [style=none] (222) at (36.5, 13.25) {};
		\node [style=none] (223) at (36.5, 15.75) {};
		\node [style=none] (224) at (35, 20.5) {};
		\node [style=none] (225) at (35, 18.5) {};
		\node [style=none] (226) at (35, 16.5) {};
		\node [style=none] (236) at (35, 21.25) {};
		\node [style=none] (237) at (35, 15.75) {};
		\node [style=none] (238) at (38, 15.75) {};
		\node [style=none] (239) at (38, 21.25) {};
		\node [font={\Huge}, style=none] (240) at (36.5, 17.5) {$\overline{X_k}$};
		\node [style=none] (241) at (32.5, 20.5) {};
		\node [style=none] (242) at (32.5, 18.5) {};
		\node [style=none] (243) at (32.5, 16.5) {};
		\node [style=none] (244) at (31.5, 20.5) {};
		\node [style=none] (245) at (31.5, 18.5) {};
		\node [style=none] (246) at (31.5, 16.5) {};
		\node [style=none] (253) at (33.25, 18.5) {$\cdots$};
		\node [style=none] (254) at (32.5, 12.5) {};
		\node [style=none] (255) at (32.5, 10.5) {};
		\node [style=none] (256) at (32.5, 8.5) {};
		\node [style=none] (257) at (31.5, 12.5) {};
		\node [style=none] (258) at (31.5, 10.5) {};
		\node [style=none] (259) at (31.5, 8.5) {};
		\node [style=none] (263) at (34, 12.5) {};
		\node [style=none] (264) at (34, 10.5) {};
		\node [style=none] (265) at (34, 8.5) {};
		\node [style=none] (266) at (33.25, 10.5) {$\cdots$};
		\node [style=none] (268) at (36.5, 19.5) {Calc.};
		\node [style=none] (271) at (30, 19.5) {Calc.};
		\node [style=none] (272) at (9.5, 20.5) {};
		\node [style=none] (273) at (9.5, 18.5) {};
		\node [style=none] (274) at (9.5, 16.5) {};
		\node [style=none] (278) at (6.75, 18.75) {\Huge{$\ket{\overline{\psi}}$}};
		\node [style=none] (279) at (9.25, 21) {};
		\node [style=none] (280) at (9.25, 16) {};
		\node [style=none] (281) at (8.25, 18.5) {};
		\node [style=none] (286) at (12, 20.5) {};
		\node [style=none] (287) at (13, 18.5) {};
		\node [style=none] (288) at (14, 16.5) {};
		\node [style=cnot ctrl] (292) at (12, 20.5) {};
		\node [style=cnot ctrl] (293) at (13, 18.5) {};
		\node [style=cnot ctrl] (294) at (14, 16.5) {};
		\node [style=cnot targ] (295) at (12, 12.5) {};
		\node [style=cnot targ] (296) at (13, 10.5) {};
		\node [style=cnot targ] (297) at (14, 8.5) {};
		\node [style=none] (298) at (12, 20.5) {};
		\node [style=none] (299) at (13, 18.5) {};
		\node [style=none] (300) at (14, 16.5) {};
		\node [style=none] (304) at (6, 10.5) {\Huge{$\ket{\overline{0}}^{\otimes k}$}};
		\node [style=none] (305) at (37, 2.5) {\Huge{$\ket{\overline{+}}^{\otimes k}$}};
		\node [style=none] (401) at (9, 12.5) {};
		\node [style=none] (402) at (9, 10.5) {};
		\node [style=none] (403) at (9, 8.5) {};
		\node [style=none] (404) at (48.5, 12.5) {};
		\node [style=none] (405) at (48.5, 10.5) {};
		\node [style=none] (406) at (48.5, 8.5) {};
		\node [style=none] (407) at (48.5, 13.25) {};
		\node [style=none] (408) at (48.5, 7.75) {};
		\node [style=none] (409) at (51.5, 7.75) {};
		\node [style=none] (410) at (51.5, 13.25) {};
		\node [font={\Huge}, style=none] (411) at (50, 10.5) {$\mathcal{D}_Z$};
		\node [style=none] (412) at (52.5, 12.5) {};
		\node [style=none] (413) at (52.5, 10.5) {};
		\node [style=none] (414) at (52.5, 8.5) {};
		\node [style=none] (415) at (51.5, 12.5) {};
		\node [style=none] (416) at (51.5, 10.5) {};
		\node [style=none] (417) at (51.5, 8.5) {};
		\node [style=none] (418) at (52.5, 13.25) {};
		\node [style=none] (419) at (52.5, 7.75) {};
		\node [style=none] (420) at (55.5, 7.75) {};
		\node [style=none] (421) at (55.5, 13.25) {};
		\node [style=none] (422) at (55.5, 12.5) {};
		\node [style=none] (423) at (55.5, 10.5) {};
		\node [style=none] (424) at (55.5, 8.5) {};
		\node [font={\Huge}, style=none] (425) at (54, 10.5) {$\mathcal{R}$};
		\node [style=none] (428) at (66, 4.5) {};
		\node [style=none] (429) at (66, 2.5) {};
		\node [style=none] (430) at (66, 0.5) {};
		\node [style=none] (432) at (46, 12.5) {};
		\node [style=none] (433) at (46, 10.5) {};
		\node [style=none] (434) at (46, 8.5) {};
		\node [style=none] (456) at (47.5, 12.5) {};
		\node [style=none] (457) at (47.5, 10.5) {};
		\node [style=none] (458) at (47.5, 8.5) {};
		\node [style=none] (459) at (46, 13.25) {};
		\node [style=none] (460) at (47.5, 13.25) {};
		\node [style=none] (461) at (47.5, 11.75) {};
		\node [style=none] (462) at (46, 11.75) {};
		\node [style=none] (463) at (46.125, 12.25) {};
		\node [style=none] (464) at (47.375, 12.25) {};
		\node [style=none] (465) at (46.75, 12.25) {};
		\node [style=none] (466) at (47.25, 13) {};
		\node [style=none] (467) at (46, 11.25) {};
		\node [style=none] (468) at (47.5, 11.25) {};
		\node [style=none] (469) at (47.5, 9.75) {};
		\node [style=none] (470) at (46, 9.75) {};
		\node [style=none] (471) at (46.125, 10.25) {};
		\node [style=none] (472) at (47.375, 10.25) {};
		\node [style=none] (473) at (46.75, 10.25) {};
		\node [style=none] (474) at (47.25, 11) {};
		\node [style=none] (475) at (46, 9.25) {};
		\node [style=none] (476) at (47.5, 9.25) {};
		\node [style=none] (477) at (47.5, 7.75) {};
		\node [style=none] (478) at (46, 7.75) {};
		\node [style=none] (479) at (46.125, 8.25) {};
		\node [style=none] (480) at (47.375, 8.25) {};
		\node [style=none] (481) at (46.75, 8.25) {};
		\node [style=none] (482) at (47.25, 9) {};
		\node [style=none] (483) at (56.5, 4.5) {};
		\node [style=none] (484) at (56.5, 2.5) {};
		\node [style=none] (485) at (56.5, 0.5) {};
		\node [style=none] (486) at (56.5, 5.25) {};
		\node [style=none] (487) at (56.5, -0.25) {};
		\node [style=none] (488) at (59.5, -0.25) {};
		\node [style=none] (489) at (59.5, 5.25) {};
		\node [font={\Huge}, style=none] (490) at (58, 2.5) {$\overline{X_1}$};
		\node [style=none] (491) at (58, 5.25) {};
		\node [style=none] (492) at (58, 7.75) {};
		\node [style=none] (493) at (56.5, 12.5) {};
		\node [style=none] (494) at (56.5, 10.5) {};
		\node [style=none] (495) at (56.5, 8.5) {};
		\node [style=none] (496) at (56.5, 13.25) {};
		\node [style=none] (497) at (56.5, 7.75) {};
		\node [style=none] (498) at (59.5, 7.75) {};
		\node [style=none] (499) at (59.5, 13.25) {};
		\node [font={\Huge}, style=none] (500) at (58, 9.5) {$\overline{Z_1}$};
		\node [style=none] (501) at (63, 4.5) {};
		\node [style=none] (502) at (63, 2.5) {};
		\node [style=none] (503) at (63, 0.5) {};
		\node [style=none] (504) at (63, 5.25) {};
		\node [style=none] (505) at (63, -0.25) {};
		\node [style=none] (506) at (66, -0.25) {};
		\node [style=none] (507) at (66, 5.25) {};
		\node [font={\Huge}, style=none] (508) at (64.5, 2.5) {$\overline{X_k}$};
		\node [style=none] (509) at (62, 12.5) {};
		\node [style=none] (510) at (62, 10.5) {};
		\node [style=none] (511) at (62, 8.5) {};
		\node [style=none] (512) at (64.5, 5.25) {};
		\node [style=none] (513) at (64.5, 7.75) {};
		\node [style=none] (514) at (63, 12.5) {};
		\node [style=none] (515) at (63, 10.5) {};
		\node [style=none] (516) at (63, 8.5) {};
		\node [style=none] (517) at (63, 13.25) {};
		\node [style=none] (518) at (63, 7.75) {};
		\node [style=none] (519) at (66, 7.75) {};
		\node [style=none] (520) at (66, 13.25) {};
		\node [font={\Huge}, style=none] (521) at (64.5, 9.5) {$\overline{Z_k}$};
		\node [style=none] (522) at (60.5, 12.5) {};
		\node [style=none] (523) at (60.5, 10.5) {};
		\node [style=none] (524) at (60.5, 8.5) {};
		\node [style=none] (525) at (59.5, 12.5) {};
		\node [style=none] (526) at (59.5, 10.5) {};
		\node [style=none] (527) at (59.5, 8.5) {};
		\node [style=none] (528) at (61.25, 10.5) {$\cdots$};
		\node [style=none] (529) at (60.5, 4.5) {};
		\node [style=none] (530) at (60.5, 2.5) {};
		\node [style=none] (531) at (60.5, 0.5) {};
		\node [style=none] (532) at (59.5, 4.5) {};
		\node [style=none] (533) at (59.5, 2.5) {};
		\node [style=none] (534) at (59.5, 0.5) {};
		\node [style=none] (535) at (62, 4.5) {};
		\node [style=none] (536) at (62, 2.5) {};
		\node [style=none] (537) at (62, 0.5) {};
		\node [style=none] (538) at (61.25, 2.5) {$\cdots$};
		\node [style=none] (539) at (64.5, 11.5) {Calc.};
		\node [style=none] (540) at (58, 11.5) {Calc.};
		\node [style=none] (541) at (38, 12.5) {};
		\node [style=none] (542) at (38, 10.5) {};
		\node [style=none] (543) at (38, 8.5) {};
		\node [style=none] (544) at (42, 0.5) {};
		\node [style=none] (545) at (43, 2.5) {};
		\node [style=none] (546) at (44, 4.5) {};
		\node [style=cnot ctrl] (547) at (42, 0.5) {};
		\node [style=cnot ctrl] (548) at (43, 2.5) {};
		\node [style=cnot ctrl] (549) at (44, 4.5) {};
		\node [style=cnot targ] (550) at (42, 8.5) {};
		\node [style=cnot targ] (551) at (43, 10.5) {};
		\node [style=cnot targ] (552) at (44, 12.5) {};
		\node [style=none] (553) at (42, 0.5) {};
		\node [style=none] (554) at (43, 2.5) {};
		\node [style=none] (555) at (44, 4.5) {};
		\node [style=none] (556) at (68, 4.5) {};
		\node [style=none] (557) at (68, 2.5) {};
		\node [style=none] (558) at (68, 0.5) {};
		\node [style=none] (586) at (40, 4.5) {};
		\node [style=none] (587) at (40, 2.5) {};
		\node [style=none] (588) at (40, 0.5) {};
	\end{pgfonlayer}
	\begin{pgfonlayer}{edgelayer}
		\draw (8.center)
			 to (9.center)
			 to (10.center)
			 to (7.center)
			 to cycle;
		\draw [style=double] (18.center) to (15.center);
		\draw [style=double] (19.center) to (16.center);
		\draw [style=double] (20.center) to (17.center);
		\draw [style=none] (24.center)
			 to (21.center)
			 to (22.center)
			 to (23.center)
			 to cycle;
		\draw [in=0, out=180] (54.center) to (52.center);
		\draw [in=180, out=0] (53.center) to (54.center);
		\draw (62.center) to (61.center);
		\draw (61.center) to (63.center);
		\draw (63.center) to (64.center);
		\draw (64.center) to (62.center);
		\draw (67.center) to (66.center);
		\draw (66.center) to (68.center);
		\draw (68.center) to (69.center);
		\draw (69.center) to (67.center);
		\draw (72.center) to (71.center);
		\draw (71.center) to (73.center);
		\draw (73.center) to (74.center);
		\draw (74.center) to (72.center);
		\draw (77.center) to (76.center);
		\draw (79.center) to (78.center);
		\draw (81.center) to (80.center);
		\draw (97.center) to (98.center);
		\draw (97.center) to (100.center);
		\draw (100.center) to (99.center);
		\draw (99.center) to (98.center);
		\draw [bend left=90, looseness=1.75] (101.center) to (102.center);
		\draw (104.center) to (103.center);
		\draw (105.center) to (106.center);
		\draw (105.center) to (108.center);
		\draw (108.center) to (107.center);
		\draw (107.center) to (106.center);
		\draw [bend left=90, looseness=1.75] (109.center) to (110.center);
		\draw (112.center) to (111.center);
		\draw (113.center) to (114.center);
		\draw (113.center) to (116.center);
		\draw (116.center) to (115.center);
		\draw (115.center) to (114.center);
		\draw [bend left=90, looseness=1.75] (117.center) to (118.center);
		\draw (120.center) to (119.center);
		\draw (154.center) to (153.center);
		\draw (153.center) to (156.center);
		\draw (156.center) to (155.center);
		\draw (155.center) to (154.center);
		\draw [style=double] (169.center) to (168.center);
		\draw [style=none] (184.center)
			 to (185.center)
			 to (182.center)
			 to (183.center)
			 to cycle;
		\draw (196.center) to (195.center);
		\draw (195.center) to (198.center);
		\draw (198.center) to (197.center);
		\draw (197.center) to (196.center);
		\draw [style=double] (223.center) to (222.center);
		\draw [style=none] (239.center)
			 to (236.center)
			 to (237.center)
			 to (238.center)
			 to cycle;
		\draw [style=double] (244.center) to (241.center);
		\draw [style=double] (245.center) to (242.center);
		\draw [style=double] (246.center) to (243.center);
		\draw (257.center) to (254.center);
		\draw (258.center) to (255.center);
		\draw (259.center) to (256.center);
		\draw [style=double] (82.center) to (4.center);
		\draw [style=double] (83.center) to (5.center);
		\draw [style=double] (84.center) to (6.center);
		\draw (60.center) to (274.center);
		\draw (59.center) to (273.center);
		\draw (58.center) to (272.center);
		\draw [in=-180, out=0] (281.center) to (279.center);
		\draw [in=0, out=-180] (280.center) to (281.center);
		\draw (292) to (295);
		\draw (293) to (296);
		\draw (294) to (297);
		\draw (401.center) to (150.center);
		\draw (402.center) to (151.center);
		\draw (403.center) to (152.center);
		\draw (28.center) to (170.center);
		\draw (29.center) to (171.center);
		\draw (30.center) to (172.center);
		\draw (263.center) to (192.center);
		\draw (264.center) to (193.center);
		\draw (265.center) to (194.center);
		\draw (408.center) to (409.center);
		\draw (409.center) to (410.center);
		\draw (410.center) to (407.center);
		\draw (407.center) to (408.center);
		\draw [style=double] (415.center) to (412.center);
		\draw [style=double] (416.center) to (413.center);
		\draw [style=double] (417.center) to (414.center);
		\draw [style=none] (421.center) to (418.center);
		\draw [style=none] (418.center) to (419.center);
		\draw [style=none] (419.center) to (420.center);
		\draw [style=none] (420.center) to (421.center);
		\draw (459.center) to (460.center);
		\draw (459.center) to (462.center);
		\draw (462.center) to (461.center);
		\draw (461.center) to (460.center);
		\draw [bend left=90, looseness=1.75] (463.center) to (464.center);
		\draw (466.center) to (465.center);
		\draw (467.center) to (468.center);
		\draw (467.center) to (470.center);
		\draw (470.center) to (469.center);
		\draw (469.center) to (468.center);
		\draw [bend left=90, looseness=1.75] (471.center) to (472.center);
		\draw (474.center) to (473.center);
		\draw (475.center) to (476.center);
		\draw (475.center) to (478.center);
		\draw (478.center) to (477.center);
		\draw (477.center) to (476.center);
		\draw [bend left=90, looseness=1.75] (479.center) to (480.center);
		\draw (482.center) to (481.center);
		\draw (487.center) to (486.center);
		\draw (486.center) to (489.center);
		\draw (489.center) to (488.center);
		\draw (488.center) to (487.center);
		\draw [style=double] (492.center) to (491.center);
		\draw [style=none] (498.center) to (499.center);
		\draw [style=none] (499.center) to (496.center);
		\draw [style=none] (496.center) to (497.center);
		\draw [style=none] (497.center) to (498.center);
		\draw (505.center) to (504.center);
		\draw (504.center) to (507.center);
		\draw (507.center) to (506.center);
		\draw (506.center) to (505.center);
		\draw [style=double] (513.center) to (512.center);
		\draw [style=none] (520.center) to (517.center);
		\draw [style=none] (517.center) to (518.center);
		\draw [style=none] (518.center) to (519.center);
		\draw [style=none] (519.center) to (520.center);
		\draw [style=double] (525.center) to (522.center);
		\draw [style=double] (526.center) to (523.center);
		\draw [style=double] (527.center) to (524.center);
		\draw (532.center) to (529.center);
		\draw (533.center) to (530.center);
		\draw (534.center) to (531.center);
		\draw [style=double] (456.center) to (404.center);
		\draw [style=double] (457.center) to (405.center);
		\draw [style=double] (458.center) to (406.center);
		\draw (434.center) to (543.center);
		\draw (433.center) to (542.center);
		\draw (432.center) to (541.center);
		\draw (547) to (550);
		\draw (548) to (551);
		\draw (549) to (552);
		\draw (586.center) to (483.center);
		\draw (587.center) to (484.center);
		\draw (588.center) to (485.center);
		\draw (422.center) to (493.center);
		\draw (423.center) to (494.center);
		\draw (424.center) to (495.center);
		\draw (535.center) to (501.center);
		\draw (536.center) to (502.center);
		\draw (537.center) to (503.center);
		\draw (428.center) to (556.center);
		\draw (429.center) to (557.center);
		\draw (430.center) to (558.center);
		\draw [style=double edge] (206.center) to (224.center);
		\draw [style=double edge] (207.center) to (225.center);
		\draw [style=double edge] (208.center) to (226.center);
		\draw [style=double edge] (511.center) to (516.center);
		\draw [style=double edge] (510.center) to (515.center);
		\draw [style=double edge] (509.center) to (514.center);
	\end{pgfonlayer}
\end{tikzpicture}